\documentclass[10pt,conference]{IEEEtran}
\usepackage{version}
\includeversion{ccs}
\excludeversion{techreport}

\usepackage{cite}
\usepackage{times}
\usepackage{epsfig,endnotes}
\usepackage{color}
\usepackage{url}
\usepackage{subfigure}
\usepackage{graphicx}
\usepackage[cmex10]{amsmath}
\usepackage{comment}
\usepackage{mdwlist}
\usepackage{balance}
\usepackage{afterpage}
\newtheorem{thm}{Theorem}
\usepackage{soul}          
\usepackage[final,inline,nomargin]{fixme}

\makeatletter
\let\@copyrightspace\relax
\makeatother

\begin{document}
\newcommand{\prateekccs}{\textcolor{black}}
\newcommand{\prateeknew}{\textcolor{black}}
\newcommand{\mattnew}{\textcolor{black}}


\title{Pisces: Anonymous Communication Using Social Networks}
\author{
\IEEEauthorblockN{Prateek Mittal}
\IEEEauthorblockA{University of California, Berkeley\\ pmittal@eecs.berkeley.edu}
\and
\IEEEauthorblockN{Matthew Wright}
\IEEEauthorblockA{University of Texas at Arlington\\ mwright@cse.uta.edu}
\and
\IEEEauthorblockN{Nikita Borisov}
\IEEEauthorblockA{University of Illinois at Urbana-Champaign\\ nikita@illinois.edu}
}
\maketitle

\begin{abstract}

The architectures of deployed anonymity systems such as Tor suffer from
two key problems that limit user's trust in these systems. First, paths
for anonymous communication are built without considering trust
relationships between users and relays in the system. Second, the
network architecture relies on a set of \prateeknew{centralized} 
servers. 
In this paper, we propose \emph{Pisces}, a decentralized protocol for anonymous
communications that leverages users' social links to build circuits
for onion routing. We argue that such an approach greatly improves the
system's resilience to attackers.

A fundamental challenge in this setting is the design of a secure
process to discover peers for use in a user's circuit. All existing
solutions for secure peer discovery leverage structured topologies and
cannot be applied to unstructured social network topologies. In Pisces,
we discover peers by using random walks in the social network graph with
a bias away from highly connected nodes to prevent a few nodes from
dominating the circuit creation process. To secure the random walks, we
leverage the \emph{reciprocal neighbor policy}: if malicious nodes
try to exclude honest nodes during peer discovery so as to improve the
chance of being selected, then honest nodes can use a tit-for-tat
approach and reciprocally exclude the malicious nodes from their routing
tables. We describe a fully decentralized protocol for enforcing this
policy, and use it to build the Pisces anonymity system.


Using theoretical modeling and experiments on real-world 
social network topologies, we show that (a) the reciprocal neighbor 
policy mitigates active attacks that an adversary can perform, 
(b) our decentralized protocol to enforce this policy 
is secure and has low overhead, and (c) the overall anonymity 
provided by our system significantly outperforms existing approaches.  
\end{abstract}  

\section{Introduction}
\label{sec:intro}

Systems for anonymous communication on the Internet, or {\em anonymity
  systems}, provide a technical means to enhance user privacy by hiding
the link between the user and her remote communicating parties (such as
websites that the user visits). Popular anonymity systems include
Anonymizer.com~\cite{anonymizer}, AN.ON~\cite{an.on, webmixes}, and
Tor~\cite{tor}. Tor is used by hundreds of thousands of
users~\cite{metrics-portal}, including journalists, dissidents,
whistle-blowers, law enforcement, and government
embassies~\cite{tor-embassies,who-uses-tor}.

Anonymity systems forward user traffic through a path (or {\em circuit})
of proxy servers. In some systems, including Tor, the proxies on the
circuit are selected from among a large number of available proxies,
each of which is supposed to be operated by a different person. An
attacker, however, might run a substantial fraction of the proxies under
different identities. He would then be able to deanonymize users whose
circuits run through his attacker-controlled proxies. Thus, the security
of the anonymity system hinges on at least some of the proxies in the
circuit being honest~\cite{or}. \mattnew{Having some means to discern
  which proxies are likely to be honest would thereby greatly enhance
  the security of the system.}

Recently, Johnson et al.\ proposed a method to incorporate
\mattnew{ {\em
  trust}---in which the user must judge which proxies are more likely to
be honest---into} a Tor-like
system~\cite{johnson:csf09,johnson:ccs11}. However, their approach
relies on
\prateeknew{central servers}, 
and offers only limited scalability (see Section~\ref{sec:background}).
Both Nagaraja~\cite{nagaraja:pet07} and Danezis et al.~\cite{drac}
describe a compelling vision for \mattnew{leveraging social
  relationships} in a decentralized (peer-to-peer) anonymity system by
building circuits over edges in a social network graph. Unfortunately,
their protocols are limited in applicability to a honest-but-curious
attacker model.
Against a more powerful Byzantine adversary, both approaches are
vulnerable to {\em route capture attacks}, in which the entire circuit
is comprised of
malicious nodes. To our knowledge, no proposed system securely
\mattnew{leverages social relationship information to improve the
  chances of attacker-free circuit construction} in a decentralized
  anonymity system.

\noindent{\bf Our contributions:}
\mattnew{In this paper, we propose to use social networks to help
  construct circuits that are more robust to compromise than any prior
  approach among decentralized anonymity systems.} We take advantage 
of the fact that, when protected from
manipulation, random walks on social network topologies are likely to
remain among the honest users~\cite{sybillimit,whanau:nsdi10}. 
We thus propose to construct random walks on a social network 
topology to select circuits in such a way that they cannot be 
manipulated by a Byzantine adversary. 
We then build circuits from these protected random walks and 
show that they provide a
very high chance for users to have an honest circuit, 
even for users who have a few social links to malicious peers.

The key challenge in this setting is to prevent the adversary from
biasing the random walk by manipulating their routing tables.
To this end, we propose the \emph{reciprocal neighbor policy}: if
malicious nodes try to exclude honest nodes during peer discovery, then
honest nodes can use a tit-for-tat approach and reciprocally exclude the
malicious nodes from their routing tables.
The policy ensures that attempts to bias the random walk towards 
malicious nodes reduce the probability of malicious nodes 
being selected as intermediate nodes in the random walk, 
nullifying the effect of the attack.  
Further, to prevent an attacker from benefiting by creating a large
clique of malicious peers in the social network, we bias random walks
away from peers with many friends.

An important contribution of our work is a technique for enforcing the
reciprocal neighbor policy in a fully decentralized fashion. We
efficiently distribute each node's current list of contacts (using
Whanau~\cite{whanau:nsdi10}) so that those contacts can verify
periodically that they are in the list. A contact that should be in the
list, but is not, can remove the node permanently from its contacts in
future time periods. Further, the list is signed by the node, so any
conflicting lists for the same time period constitute proof that the
node is cheating.
Using this policy, we %
design %
Pisces, a decentralized
anonymity system that uses %
random walks on social networks to
take advantage of users' trust relationships without being exposed to
circuit manipulation.

We demonstrate through %
theoretical analysis,
simulation, and experiments, that our application of the reciprocal neighbor policy provides
good deterrence against active attacks. We also show that our
distributed design provides robust enforcement of this policy, with
manageable overhead for distributing and checking contact
lists. 
Finally, using real world social network topologies, %
we show that Pisces provides significantly higher
anonymity than existing approaches.
Compared with %
decentralized approaches that do not leverage social networks
 (like %
ShadowWalker%
~\cite{shadowwalker}), Pisces provides up to six bits higher
entropy in a single communication round. Compared with the naive
strategy of using conventional random walks over social networks (as in
the Drac system~\cite{drac}), Pisces provides twice the %
entropy over 100 communication rounds.

\section{Background and Related Work}
\label{sec:background}

The focus of this work is on low-latency anonymity systems that can be
used for interactive traffic such as Web browsing and instant
messaging. Low-latency anonymity systems aim to defend against a partial
adversary who can compromise or monitor only a fraction of links in the
system. Most of these systems rely on \emph{onion
  routing}\cite{onion-routing} for anonymous communication. Onion
routing enables anonymous communication by using a sequence of relays as
intermediate nodes to forward traffic. Such a sequence of relays is
referred to as a \emph{circuit}. A key property of onion routing is that
each relay on the circuit only sees the identity of the previous hop and
the next hop, but no single relay can link both the initiator and the
destination of the communication.

\subsection{Centralized/semi-centralized approaches}

Most deployed systems for anonymous communication have a
centralized or semi-centralized architecture, including
Anonymizer~\cite{anonymizer}, AN.ON~\cite{an.on}, Tor~\cite{tor},
Freedom~\cite{freedom}, Onion Routing~\cite{onion-routing}, and
I2P~\cite{i2p}.
Anonymizer.com~\cite{anonymizer} is effectively a centralized proxy
server with a single point of control. If the proxy server becomes
compromised or is subject to subpoena, the privacy provided by the
system would be lost. AN.ON~\cite{an.on} distributes the trust among
three independently-operated servers; again, the compromise of just a
few nodes suffices to undermine the entire system. Both Anonymizer.com
and AN.ON are prone to %
flooding-based denial of service
attacks. Furthermore, with both systems, it may be possible to eavesdrop
on the server(s) and use end-to-end timing attacks (such
as~\cite{wang:ccs05, dsss, rainbow}) to substantially undermine the
privacy of all users.

Tor is a widely used anonymous communication system, serving roughly
500,000 users~\cite{metrics-users} and carrying terabytes of traffic
each day~\cite{metrics-portal}. Tor is substantially more distributed
than either Anonymizer.com or AN.ON, with users building circuits from
among about 3,000 proxy nodes (onion routers) as of May
2012~\cite{torstatus-blutmagie}. This helps to protect against direct
attacks and eavesdropping on the entire system.
Tor relies on trusted entities called {\em directory authorities} to
maintain up-to-date information about all relays that are online in the
form of a network consensus database.
Freedom~\cite{freedom} and I2P~\cite{i2p} 
also use such centralized directory servers. Users 
download the full database, and then locally select random relays for
anonymous communication. Clients download this database every three
hours to handle relay churn.

Although Tor has a more distributed approach than any other deployed
system, there are several shortcomings with its architecture. First, Tor
does not leverage a user's trust relationships for building circuits. In
Tor, an attacker could volunteer a set of proxy nodes under different
identities and then use these nodes to compromise the anonymity of
circuits going through them. Leveraging trust relationships has been
shown to be useful for improving anonymity against such an attacker in
Tor~\cite{johnson:csf09,johnson:ccs11}.
Second, the trusted directory authorities are attractive targets for
attack; in fact, some directory authorities were recently found to have
been compromised~\cite{tor-dir-auth-compromise}. Finally, the
requirement for all users to maintain global information about all
online relays becomes a scalability bottleneck. McLachlan et
al.~\cite{torsk} showed that under reasonable growth projections, the
Tor network could be spending more bandwidth to maintain this global
system view than for the core task of relaying anonymous communications.
The recent proposal for PIR-Tor~\cite{pir-tor} might address the
networking scalability issues, %
but it does not mitigate the basic trust and denial of service
issues in a centralized approach.

\subsection{Incorporating social trust}

The importance of leveraging social network trust relationships to
improve the security and privacy properties of systems has been
recognized by a large body of previous work~\cite{advogato,sybilguard,
  sybillimit,sybilinfer,srdhtr,whanau:nsdi10,x-vine,drac,nagaraja:pet07,
  johnson:csf09,johnson:ccs11,membership-concealing}. Recently, Johnson
et al.\ proposed a method to incorporate trust into a Tor-like
system~\cite{johnson:csf09,johnson:ccs11}.  However, their approach
relies on \prateeknew{central servers}, %
and offers only limited
scalability. Nagaraja~\cite{nagaraja:pet07} and Danezis et
al.~\cite{drac} have both proposed anonymity systems over social
networks.
However, both approaches assume an honest-but-curious attack model and
are vulnerable to route capture attacks. In particular, without the
security of the reciprocal neighbor policy used in Pisces, a random walk
on the social graph that goes to an attacker-controlled peer at any step
can be controlled by the attacker for the remainder of the walk.

Designing anonymity systems that are aware of users' trust relationships
is an important step towards defending against the {\em Sybil
  attack}~\cite{sybil}, in which a single entity in the network (the
attacker) can emulate the behavior of multiple identities and violate
security properties of the system. Mechanisms like
SybilGuard~\cite{sybilguard}, SybilLimit~\cite{sybillimit}, and
SybilInfer~\cite{sybilinfer} aim to leverage social network trust
relationships to bound the number of Sybil identities any malicious
entity can emulate. These mechanisms are based on the observation that
it is costly for an adversary to form trust relationships (also known as
\emph{attack edges}) with honest nodes.  When the adversary performs a
Sybil attack, he can create an arbitrary number of edges between Sybil
identities and malicious entities, but cannot create trust relationships
between Sybil identities and honest users. Thus, a Sybil attack in
social networks creates two regions in the social network graph, the
honest region and the Sybil region, with relatively few edges between
them; i.e., the graph features a {\em small cut}. This cut can be used
to detect and mitigate the Sybil attack.

Recent work has challenged the assumption that it is costly for an
attacker to create attack edges with honest nodes in friendship
graphs~\cite{wilson:eurosys09,boshmaf:acsac11,bilge:www09,irani:dimva11},
and proposed the use of interaction graphs as a more secure realization
of real world social trust. In this work, we will evaluate Pisces with
both %
friendship graphs as well as topologies based on
interaction graphs. Other mechanisms to infer the strength of ties
between users~\cite{gilbert-karahalios:chi09} may also be helpful in
creating resilient social graphs, but these are not the focus of this
paper.

\subsection{Decentralized and peer-to-peer approaches}

A number of distributed directory services for anonymous communication 
have been designed using a P2P approach; most have serious problems that 
prevent them from being deployed. We point out these issues briefly here.
First we note that the well-known Crowds system, which was the first P2P
anonymity system, uses a centralized directory service~\cite{crowds}
and thus is not fully P2P. The Tarzan system proposes a gossip-based
distributed directory service that does not scale well beyond 10,000
nodes~\cite{tarzan}. 

MorphMix was the first scalable P2P anonymity system~\cite{morphmix}. It
uses random walks on unstructured topologies for circuit construction
and employs a {\em witness} scheme that aims to detect routing table
manipulation. The detection mechanism can be bypassed by an attacker who
is careful in his choices of fake routing
tables~\cite{tabriz:pet06}. With or without the evasion technique, the
attacker can manipulate routing tables to capture a substantial fraction
of circuits. Despite over a decade of research, decentralized mechanisms
to secure random walks in unstructured topologies have been an open
problem. Pisces overcomes this problem by having peers sign their
routing tables for a given time slot and ensuring that nodes observe
enough copies to detect cheaters quickly, before many route captures can
occur, and with certainty.

Recently, several protocols have been proposed using P2P systems built
on distributed hash tables (DHTs), including AP3~\cite{ap3},
Salsa~\cite{salsa}, NISAN~\cite{nisan}, and Torsk~\cite{torsk}. These
four protocols are vulnerable to information leak
attacks~\cite{mittal:ccs08, wang:ccs10}, as the lookup process and
circuit construction techniques expose information about the requesting
peer's circuits. These attacks can lead to users being partially or
completely deanonymized. ShadowWalker, which is also based on a
structured topology, employs random walks on the DHT topology for
circuit construction~\cite{shadowwalker}. The routing tables are checked
and signed by {\em shadow nodes} such that both a node in the random
walk and all of its shadows would need to be attackers for a route
capture attack to succeed. The protocol was found to be partially
broken, but then also fixed, by Schuchard et
al.~\cite{schuchard:wpes10}. Despite it not being seriously vulnerable
to attacks as found against other protocols, ShadowWalker remains
vulnerable to a small but non-trivial fraction of route captures
(roughly the same as Tor for reasonable parameters). 
As we show in
Section~\ref{sec:evaluation}, Pisces's use of social trust means that it
can outperform ShadowWalker (and Tor) for route captures when the
attacker has a bounded number of attack edges in the social network.

Other than these approaches, Mittal et al.~\cite{mittal:hotsec10}
briefly considered the use of the reciprocal neighbor policy for
anonymous communication. However, their protocol is only applicable to
constant degree topologies, and utilizes central points of
trust. Moreover, their evaluation was very preliminary. In this paper,
we present a complete design for a decentralized anonymity system based
on the reciprocal neighbor policy. Since our design is not limited
to constant degree topologies, we explore the advantages that come from
applying the technique to unstructured social network graphs. We also
present the first fully decentralized protocols for achieving these
policies and present analysis, simulation, and experimentation results
demonstrating the security and performance properties of the Pisces
approach.

\section{Pisces Protocol}
\label{sec:protocol}

In this section, we first describe our design goals, threat model, and
system model. We then outline the core problem of securing random walks
and describe the role of the reciprocal neighbor policy in solving the
problem in the context of social networks. Finally, we explain how
Pisces %
\prateekccs{securely implements this policy.}

\subsection{Design goals}

We now present our key design goals for our system.

\noindent {\em 1.\ Trustworthy anonymity:} we target an architecture 
that is able to leverage a user's social trust relationships to improve 
the security of anonymous communication. Current mechanisms 
for incorporating social trust are either centralized or are limited 
in applicability to an honest-but-curious attacker model.

\noindent {\em 2.\ Decentralized design:} the design should not 
have any central entities. Central entities are attractive targets 
for attackers, in addition to being a single point of failure for 
the entire system. \prateekccs{The design should also mitigate route capture 
and information leak attacks.} 

\noindent {\em 3.\ Scalable anonymity:} the design should be able to
scale to millions of users and relays with low communication
overhead. Since anonymity is defined as the state of being
unidentifiable in a group~\cite{anon_terminology}, architectures that
can support millions of users provide the additional benefit of
increasing the overall anonymity of users.

\subsection{Threat model}
In this work, we consider a colluding adversary who can launch Byzantine
attacks against the anonymity system. The adversary can perform passive
attacks such as logging information for end-to-end timing
analysis~\cite{levine:fc04}, as well as active attacks such as deviating
from the protocol and selectively denying service to some
circuits~\cite{borisov:ccs07}.
We assume the existence of effective mechanisms to defend against the
Sybil attack, such as those based on social
networks~\cite{sybillimit,sybilinfer}. Existing Sybil defense mechanisms
are not perfect and allow the insertion of a bounded number of Sybil
identities in the system. They also require the number of attack edges
to be bounded by $g=O(\frac{h}{\log h})$, where $h$ is the number of
honest nodes in the system. We use this as our primary threat model. For
comparative analysis, we also evaluate our system under an ideal Sybil
defense that does not allow the insertion of any Sybil identities.

\subsection{System Model and Assumptions}

Pisces is a fully decentralized protocol and does not assume any PKI
infrastructure. Each node generates a local public-private key pair. An
identity in the system equates to its public key. Existing Sybil defense
mechanisms can be used to validate node identities. %
We assume that the identities in
the system can be blacklisted; i.e., an adversary node cannot whitewash
its identity by rejoining the system with a different public key. This
is a reasonable assumption, since (a) mechanisms such as
SybilInfer/SybilLimit only allow the insertion of a bounded number of
Sybil identities, and (b) replacing deleted attack edges is expensive
for the attacker, particularly in a social network graph based on
interactions. We assume loose time synchronization amongst
nodes. Existing services such as NTP~\cite{ntp} can provide time
synchronization on the order of hundreds of milliseconds in wide area
networks~\cite{ntpstudy}.

Finally, in this work, we will leverage mechanisms for building
efficient communication structures over unstructured social networks,
such as Whanau~\cite{whanau:nsdi10} and X-Vine~\cite{x-vine}. These
mechanisms embed a structure into social network topologies to provide a
distributed hash table for efficient communication~\cite{chord,pastry}.
In particular, we use Whanau, since it provides the best security
guarantees amongst the current state of art. Whanau guarantees that,
with high probability, it is possible to securely look up any object
present in the DHT. It is important to point out that Whanau only
provides \emph{availability}, but not
\emph{integrity}~\cite{whanau:nsdi10}. This means that if a user
performs redundant lookups for a single key, multiple conflicting
results may be returned; Whanau guarantees that the correct result will
be included in the returned set, but leaves the problem of identifying
which result is correct to the application layer. Therefore, Whanau
cannot be used in conjunction with current protocols that provide
anonymous communication using structured topologies~\cite{shadowwalker},
since these protocols require integrity guarantees from the DHT layer
itself. We emphasize that the only property we assume from Whanau is
secure routing; in particular, we do not assume any privacy or anonymity
properties in its lookup mechanisms~\cite{mittal:ccs08,wang:ccs10}.
\subsection{Problem Overview}

Random walks are an integral part of many distributed anonymity systems,
from Tor~\cite{tor} to ShadowWalker~\cite{shadowwalker}. In a random
walk \-based circuit construction, an initiator $I$ of the random walk
first selects a random node $A$ from its neighbors in some topology (in
our case, the social network graph). The initiator sets up a single-hop
circuit with node $A$ and uses the circuit to download a list of node
$A$'s neighbors (containing the IP addresses and public keys of
neighbors). Node $I$ can then select a random node $B$ from the
downloaded list of node $A$'s neighbors and extend the circuit through
$A$ onto node $B$. This process can be repeated to set up a circuit of
length $l$.

Random walks are vulnerable to active route capture attacks in which an
adversary biases the peer discovery process towards colluding malicious
nodes. First, malicious nodes can exclude honest nodes from their
neighbor list to bias the peer discovery process. Second, malicious
nodes can modify the public keys of honest nodes in their neighbor
list. When a initiator of the random walk extends a circuit from a
malicious node to a neighboring honest node, the malicious node can
simply emulate the honest neighbor. The malicious node can repeat this
process for further circuit extensions as well. Finally, the malicious
nodes can add more edges between each other in the social network
topology to increase the percentage of malicious nodes in their neighbor
lists.
To secure the random walk process, we use a reciprocal
neighbor policy that 
limits the benefit to the attacker of attempting to bias the random
walks. We propose a protocol that securely realizes this policy
through detection of violations.

\subsection{Reciprocal Neighbor Policy}
We 
now discuss the key primitive we leverage for securing random walks, the
\emph{reciprocal neighbor policy}. The main idea
of this policy is to consider undirected versions of structured or
unstructured topologies and then entangle the routing tables of
neighboring nodes with each other. In other words, if a malicious node
$X$ does not correctly advertise an honest node $Y$ in its neighbor
list, then $Y$ also excludes $X$ from its neighbor list in a tit-for-tat
manner.
The reciprocal neighbor policy ensures that route capture attacks
based on incorrect advertisement of honest nodes during random walks
serves to partially isolate malicious nodes behind a small cut in the
topology, reducing the probability that they will be selected in a
random walk. In particular, this policy mitigates the first two types of
route capture attacks described above, namely the exclusion of honest
nodes and the modification of honest nodes' public keys. However, the
adversary can still bias the peer discovery process by simply inserting
a large number of malicious nodes to its routing tables. Thus, the
reciprocal neighbor policy described so far %
would only be effective for topologies
in which node degrees are bounded and homogeneous, such as structured
peer-to-peer topologies like Chord~\cite{chord} and
Pastry~\cite{pastry}. However, node degrees in unstructured social
network topologies are highly heterogeneous, presenting an avenue for
attack.
{\bf Handling the node degree attack: } Addition of edges amongst
colluding malicious nodes in a topology increases the probability that a
malicious node is selected in a random walk. To prevent this node degree
attack, we propose to perform random walks using the Metropolis-Hastings
modification~\cite{metropolis,MetropolisHastings} --- the transition
matrix used for our random walks is as follows:

\begin{ccs}
\vspace{-10pt}
\end{ccs}
\begin{equation}
P_{ij}  = 
\begin{cases}
 \min(\frac{1}{d_i}, \frac{1}{d_j})& \text{ if }i \rightarrow j\text{ is an edge in G} \\
1-\sum_{k \neq i} P_{ik} & \text{ if } j =  i \\
0&  \text{  otherwise}
\end{cases}
\end{equation}

\noindent where $d_i$ denotes the degree of vertex $i$ in $G$. Since the
transition probabilities to neighbors may not always sum to one, nodes
add a self loop to the transition probabilities to address this. The
Metropolis-Hastings modification ensures that attempts to add malicious
nodes in the neighbor table decreases the probability of malicious nodes
being selected in a random walk.
We will show that the Metropolis-Hastings modification along with
reciprocal neighbor policy is surprisingly effective at mitigating
active attacks on random walks. A malicious node's attempts to bias the
random walk process by launching route capture attacks reduce its own
probability of getting selected as an intermediate node in future random
walks, nullifying the effect of the attack.

\subsection{Securing Reciprocal Neighbor Policy}

We now present our protocol for securely implementing the 
reciprocal neighbor policy.

{\bf Intuition:} Our key idea is to have neighbors periodically check
each other's neighbor lists. Suppose that node $X$ and node $Y$ are
neighbors. If node $X$'s neighbor list doesn't include node $Y$, then
the periodic check will reveal this and enable node $Y$ to implement the
tit-for-tat removal of node $X$ from its routing table. Additionally,
the neighbor lists can be signed by each node so that a dishonest node
can be caught with two different, signed lists and blacklisted. To
handle churn, we propose that all nodes keep their neighbor lists static
for the duration of a regular interval ($t$) and update the list with
joins and leaves only between intervals. Here we rely on our assumption
of loose time synchronization. The use of static neighbor lists ensures
that we can identify conflicting neighbor lists from a node for the same
time interval, which would be a clear indication of malicious behavior.
The duration of the time interval for which the lists remain static
determines the trade-off between the communication overhead for securing
the reciprocal neighborhood policy and the unreliability of circuit
construction due to churn.

{\bf Setting up static neighbor list certificates:} A short time prior
to the beginning of a new time interval, each node sets up a new
neighbor list that it will use in the next time interval:

\begin{enumerate*}
  \item \emph{Liveness check}: In the first round, nodes exchange
    messages with their trusted neighbors to check for liveness and
    reciprocal trust. A reciprocal exchange of messages ensures that
    both neighbors are interested in advertising each other in the next
    time interval (and are not in each other's local blacklists). Nodes
    wait for a time duration to receive these messages from all
    neighbors, and after the timeout, construct a preliminary version of
    their next neighbor list, comprising node identities of all nodes
    that responded in the first communication round.

  \item \emph{Degree exchange}: Next, the nodes broadcast the length of
    their preliminary neighbor list to all the neighbors. This step is
    important since Metropolis-Hastings random walks require node
    degrees of neighboring nodes to determine their transition
    probabilities.

  \item \emph{Final list}: After receiving these broadcasts from all the
    neighbors, a node creates a final neighbor list and digitally signs
    it with its private key. The final list includes the IP address,
    public key, and node degree of each neighbor, as well as the time
    interval for the validity of the list.  Note that a neighbor may go
    offline between the first and second step, before the node has a
    chance to learn its node degree, in which case it can simply be
    omitted from the final list.

  \item {\em Local integrity checks}: At the beginning of every new time
    interval, each node queries all its neighbors and downloads their
    signed neighbor lists.  When a node $A$ receives a neighbor list
    from $B$, it performs local integrity checks, verifying that $B$'s
    neighbor entry for $A$ contains the correct IP address, public key,
    and node degree. Additionally, it verifies that the length of the
    neighbor list is at most as long as was broadcast in Phase 2. (Note
    that intentionally broadcasting a higher node degree is
    disadvantageous to a $B$, as it will reduce the transition
    probability of it being chosen by a random walk). If any local
    integrity checks fails, $A$ places $B$ in its permanent local
    blacklist, severing its social trust relationship with $B$ and
    refusing all further communication. If all the checks succeed, then
    these neighbor lists serve as a cryptographic commitment from these
    nodes--the presence of any conflicting neighbor lists for the same
    time interval issued by the same node is clear evidence of
    misbehavior. 
    
    If $B$'s neighbor list omits $A$ entirely, or if $B$ simply refuses
    to send its neighbor list to $A$, $B$ is placed on a temporary
    blacklist, and $A$ will refuse further communication with $B$ for
    the duration of the current time period, preventing any circuits
    from being extended from $A$ to $B$. (Effectively, $A$ performs a
    selective denial-of-service against $B$; see
    Section~\ref{sec:seldos} for more discussion of this.)  The
    blacklist only lasts for the duration of the current round, since
    the omission could have resulted from a temporary communication
    failure.
\end{enumerate*}

{\bf Duplicate detection:} Next, we need to ensure that $B$ uses the
same neighbor list during random walks as it presented to its neighbors.
Our approach is to use the Whanau DHT to check for the presence of
several conflicting neighbor lists signed by the same node for the same
time period. After performing the local checks, $A$ will store a copy of
$B$'s signed neighbor list in the Whanau, using $B$'s identity (namely,
its public key) as the DHT key.  Then, when another node $C$ performs a
random walk that passes through $B$, it will receive a signed neighbor
list from $B$. It will then perform a lookup in the DHT for any stored
neighbor lists under $B$'s key. If it discovers a different list for the
same period with a valid signature, then it can notify $B$'s neighbors
about the misbehavior, causing them to immediately blacklist $B$.

One challenge is that the Whanau lookups are not anonymous and may
reveal to external observers the fact that $C$ is performing a random
walk through $B$. This information leak, linking $C$ and $B$, can then
be used to break $C$'s anonymity~\cite{mittal:ccs08,wang:ccs10}. To
address this problem, we introduce the concept of \emph{testing} random
walks that are not actually used for anonymous communication but are
otherwise indistinguishable from regular random walks. Whanau lookups to
check for misbehavior are performed for
testing random walks only, since information leaks in that case will not
reveal private information. The lookups are performed after the random
walk to ensure that testing walks and the regular walks cannot be
distinguished.
If each node performs a small number of testing walks within a each time
period, any misbehavior will be detected with high probability.

{\bf Blacklisting:} When $C$ detects a conflicting neighbor list issued
by $B$, it immediately notifies all of $B$'s neighbors (as listed in the
neighbor list stored in the DHT), presenting the two lists as evidence
of misbehavior. $B$'s neighbors will thereafter terminate their social
relationships with $B$, blacklisting it. Note, however, that the two
conflicting lists form incontrovertible evidence that $B$ was behaving
maliciously, since honest nodes \emph{never} issue two neighbor lists in
a single time interval. This evidence can be broadcast globally to
ensure that \emph{all} nodes blacklist $B$, as any node can verify the
signatures on the two lists, and thus $B$ will not be able to form
connections with any honest nodes in the system.  Moreover, honest nodes
will know not to select $B$ in any random walk, effectively removing it
from the social graph entirely.

{\bf Proactive vs.\ reactive security:} Our system relies on detecting
malicious behavior and blacklisting nodes. Thus, as described so far,
Pisces provides reactive security. To further strengthen random walk
security in the scenario when the adversary is performing route capture
for the first time, we propose an extension to Pisces that aims to
provide proactive security. We propose a \emph{discover but wait}
strategy, in which users build circuits for anonymous communication, but
impose a delay between building a circuit and actually using it for
anonymous communication.  If misbehavior is detected by a testing random
walk within the delay period, the circuit will be terminated as $B$'s
neighbors blacklist it; otherwise, if a circuit survives some timeout
duration, then it can be used for anonymous communication.

{\bf Performance optimization:} Using all hops of a random walk for
anonymous communication has significant performance limitations. First,
the latency experienced by the user scales linearly with the random walk
length. Second, long circuit lengths reduce the overall throughput that
a system can offer to a user. Inspired by prior
work~\cite{shadowwalker}, we propose the following performance
optimization. Instead of using all hops of a random walk for anonymous
communication, the initiator can use the random walk as a peer discovery
process, and leverage the $k$th hop and the last hop to build a two-hop
circuit for anonymous communication.  In our evaluation, we find that
values of $k$ that are close to half the random walk length provide a
good trade-off between anonymity and performance.

\section{Evaluation}
\label{sec:evaluation}
In this section, we evaluate Pisces with theoretical analysis as well as
experiments using real-world social network topologies. In particular,
we (a) show the security benefits provided by the reciprocal neighbor policy, (b) evaluate the
security, performance, and overhead of our protocol that implements the
policy, and (c) evaluate the overall anonymity provided by Pisces.
We consider four datasets for our experiments, which were processed in a
manner similar to the evaluation done in SybilLimit~\cite{sybillimit}
and SybilInfer~\cite{sybilinfer}: (i) {\em a Facebook friendship graph
  from the New Orleans regional network~\cite{vishwanath-wosn09}},
containing 50,150 nodes and 772,843 edges; (ii) {\em a Facebook wall
  post interaction graph from the New Orleans regional
  network~\cite{vishwanath-wosn09}}, containing 29,140 users and 161,969
edges;
(iii) {\em a Facebook interaction graph from a moderate-sized regional
  network~\cite{wilson:eurosys09}}, containing about 380,564 nodes and
about 3.24 million edges; (iv) {\em a Facebook friendship graph from a
  moderate-sized regional network~\cite{wilson:eurosys09}}, containing
1,033,805 nodes and about 13.7 million edges. %

\begin{figure*}[ht]
\centering
\mbox{
\hspace{-0.2in}
\hspace{-0.12in}
\begin{tabular}{c}
\psfig{figure=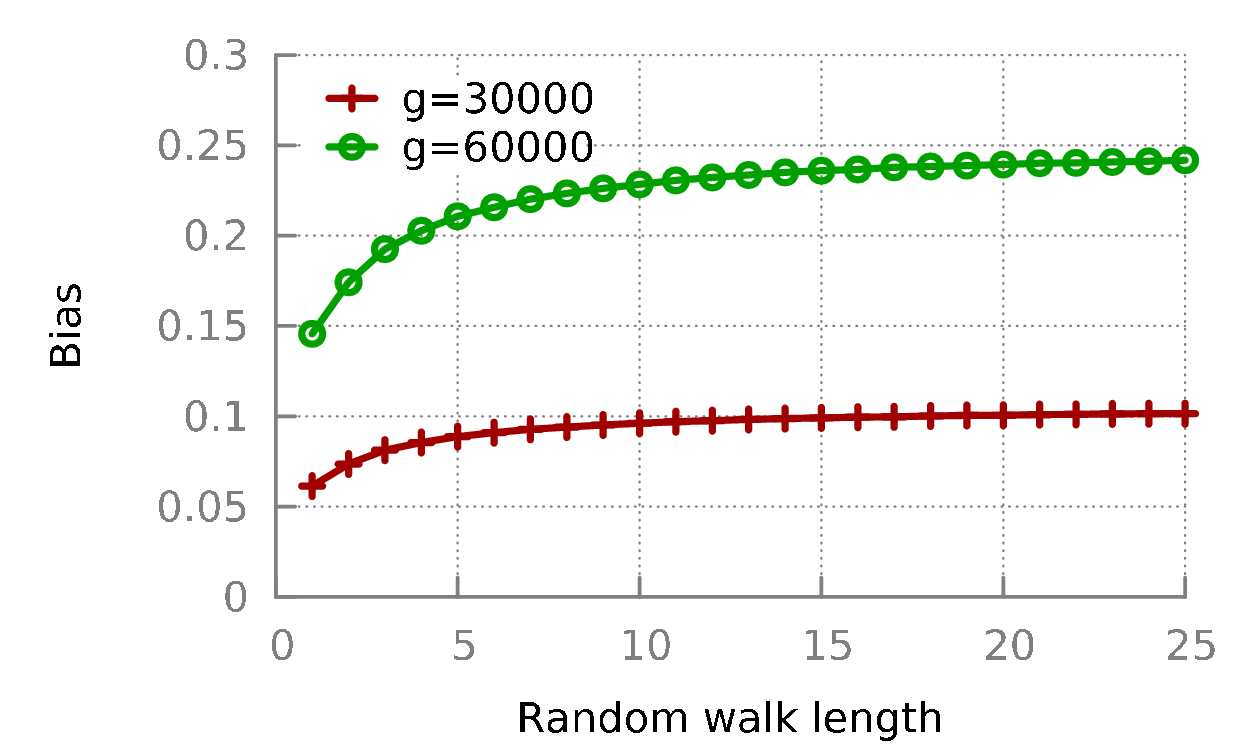,width=0.33 \textwidth}\vspace{-0.00in}\\
{(a)}
\end{tabular}
\begin{tabular}{c}
\psfig{figure=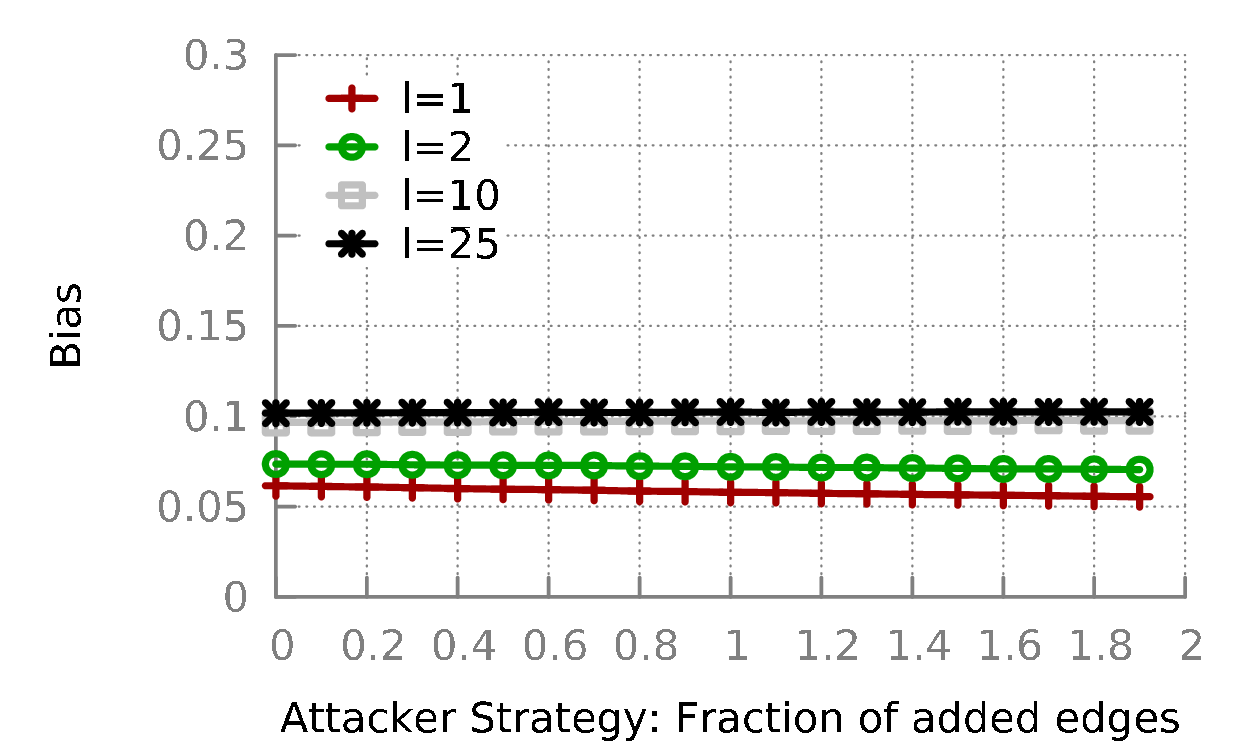,width=0.33 \textwidth}\vspace{-0.00in}\\
{(b)}
\end{tabular}
\hspace{-0.2in}
\begin{tabular}{c}
\psfig{figure=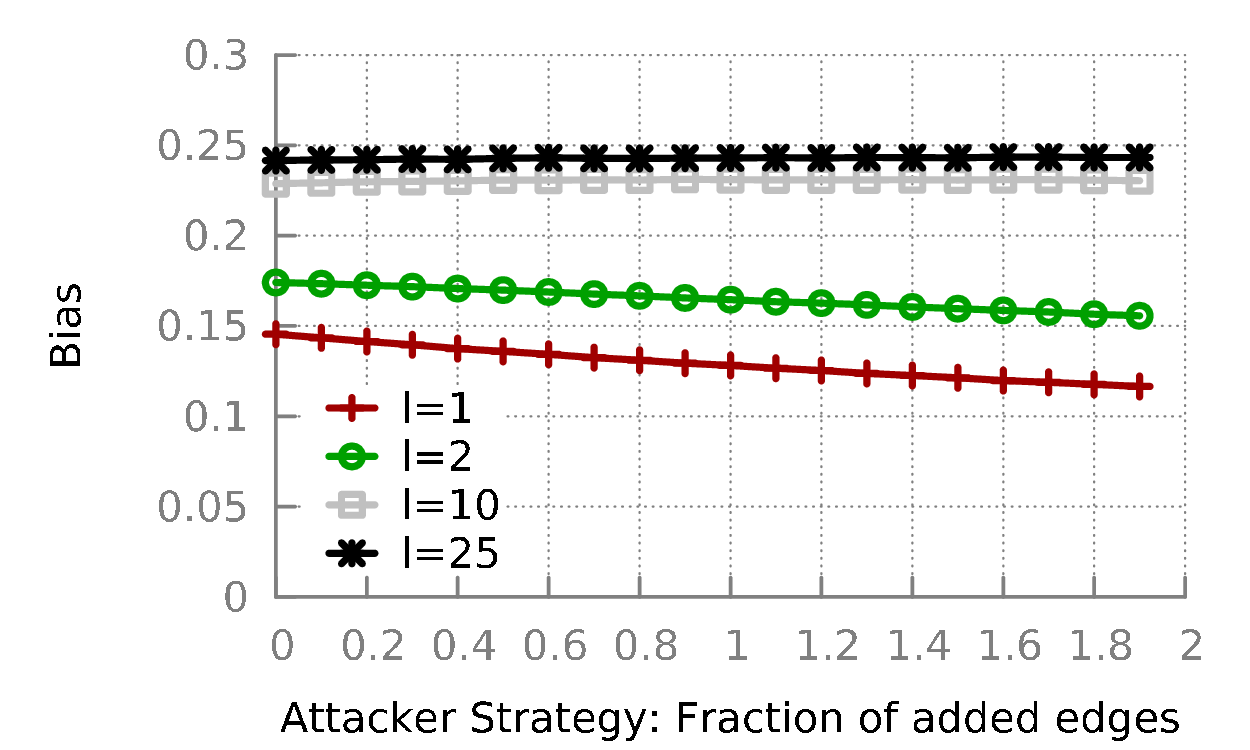,width=0.33\textwidth}\vspace{-0.00in}\\
{(c)}
\end{tabular}
}
\begin{ccs}
\vspace{-10pt}
\end{ccs}
\caption{{\em Probability of the $l$'th hop being compromised (Sampling
    Bias), under an increasing node degree attack [Facebook wall post graph]} (a) Without attack
  (b) g=30000 attack edges, (c) g=60000 attack edges.  For short random
  walks, this is a losing strategy for the adversary. For longer random
  walks, the adversary does not gain any advantage.}
\label{fig:degree-attack}
\begin{ccs}
\vspace{-10pt}
\end{ccs}
\end{figure*}
\begin{techreport}
\begin{figure*}[htp]
\centering
\mbox{
\hspace{-0.2in}
\hspace{-0.12in}
\begin{tabular}{c}
\psfig{figure=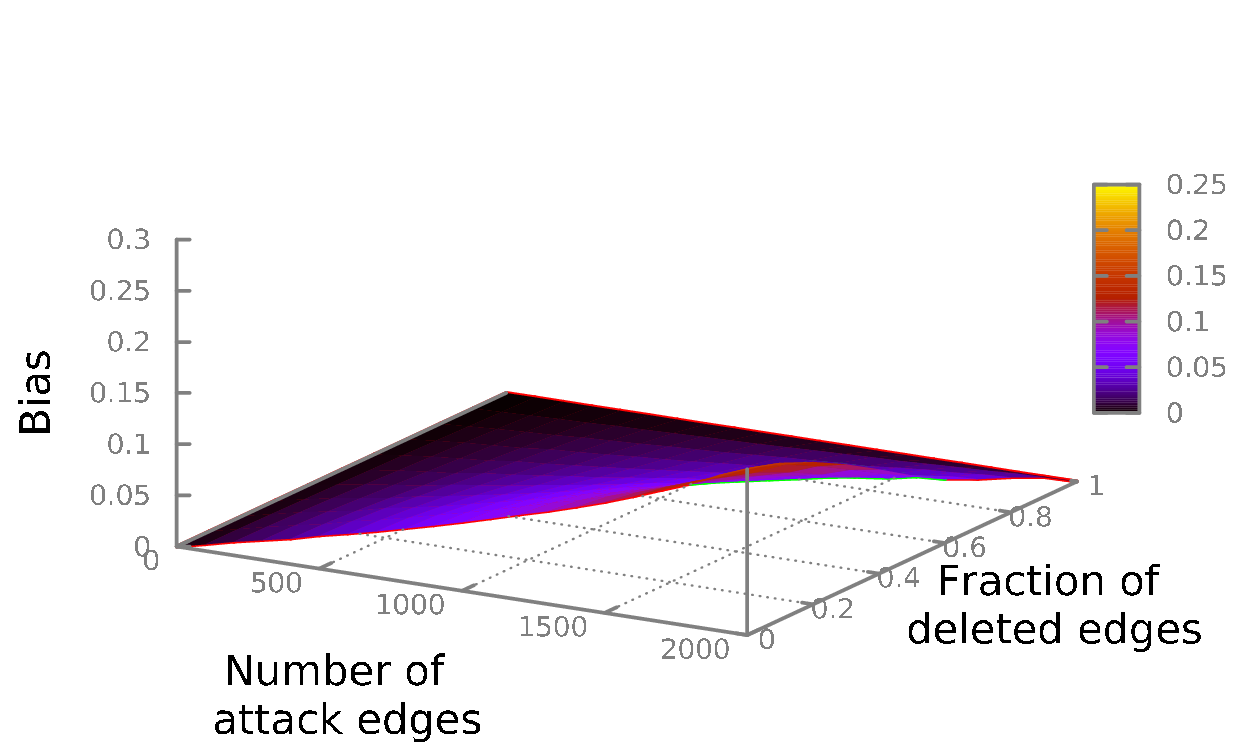,width=0.33\textwidth}\vspace{-0.00in}\\
{(a)}
\end{tabular}
\begin{tabular}{c}
\psfig{figure=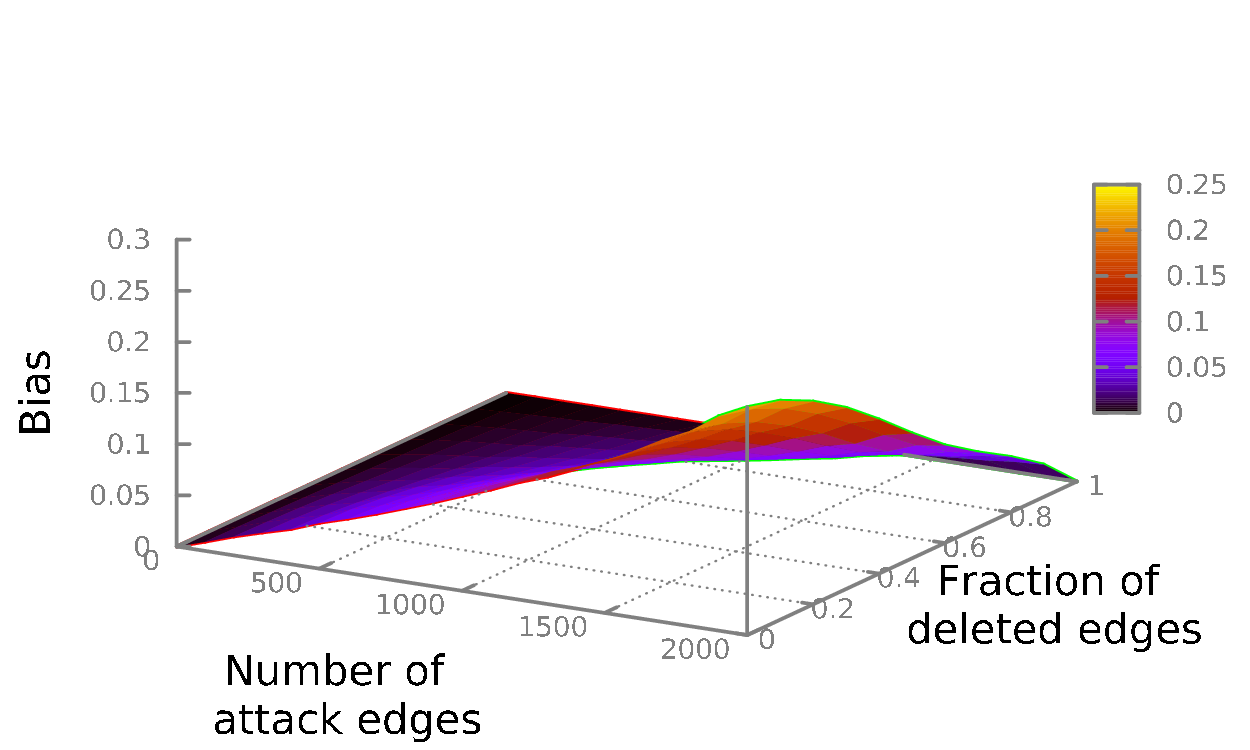,width=0.33\textwidth}\vspace{-0.00in}\\
{(b)}
\end{tabular}
\hspace{-0.2in}
\begin{tabular}{c}
\psfig{figure=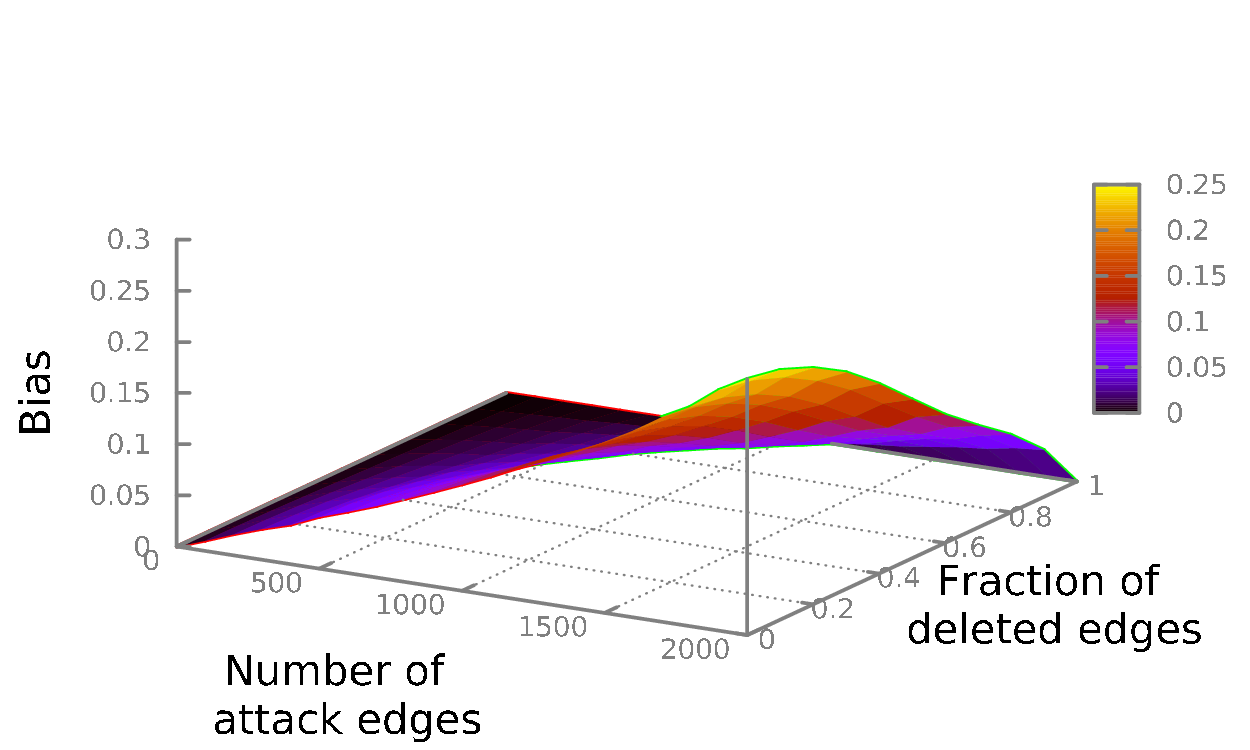,width=0.33\textwidth}\vspace{-0.00in}\\
{(c)}
\end{tabular}
}
\caption{{\em Probability of the $l$'th hop being compromised (Sampling
    Bias) under a route capture attack} (a) $l=1$, (b) $l=3$, (c)
  $l=6$. As more edges to the honest nodes are removed, the attacker's
  loss is higher. Note that the impact is very high on small length
  random walks, but gets smaller for longer length random walks. }
\label{fig:local-blacklisting}
\end{figure*}
\end{techreport}
\begin{figure*}[ht]
\centering
\mbox{
\hspace{-0.2in}
\hspace{-0.12in}
\begin{tabular}{c}
\psfig{figure=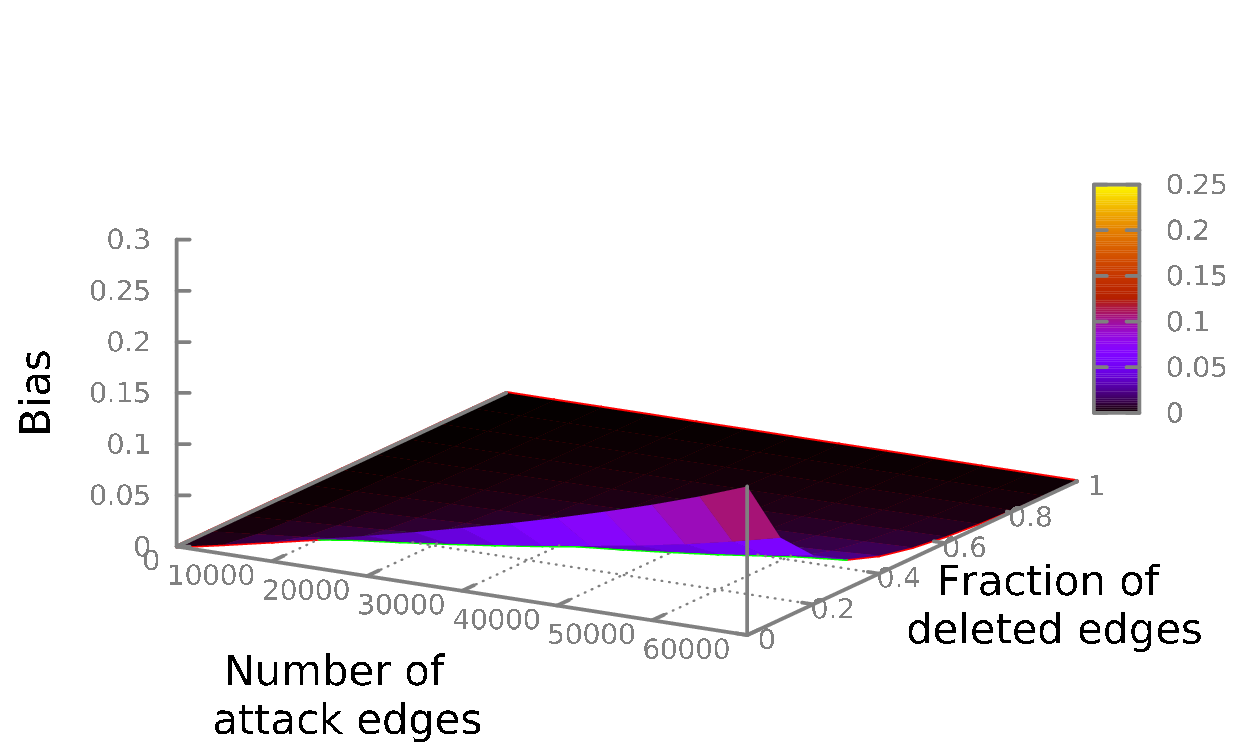,width=0.33\textwidth}\vspace{-0.00in}\\
{(a)}
\end{tabular}
\begin{tabular}{c}
\psfig{figure=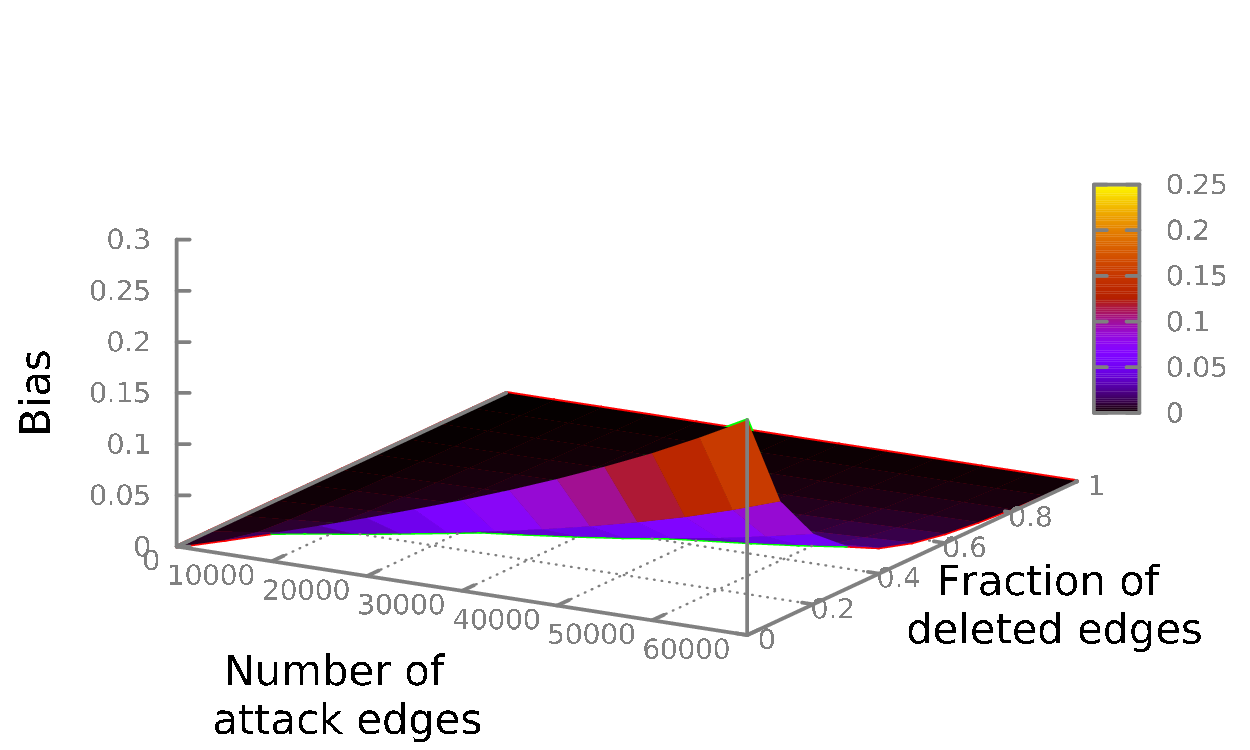,width=0.33\textwidth}\vspace{-0.00in}\\
{(b)}
\end{tabular}
\hspace{-0.2in}
\begin{tabular}{c}
\psfig{figure=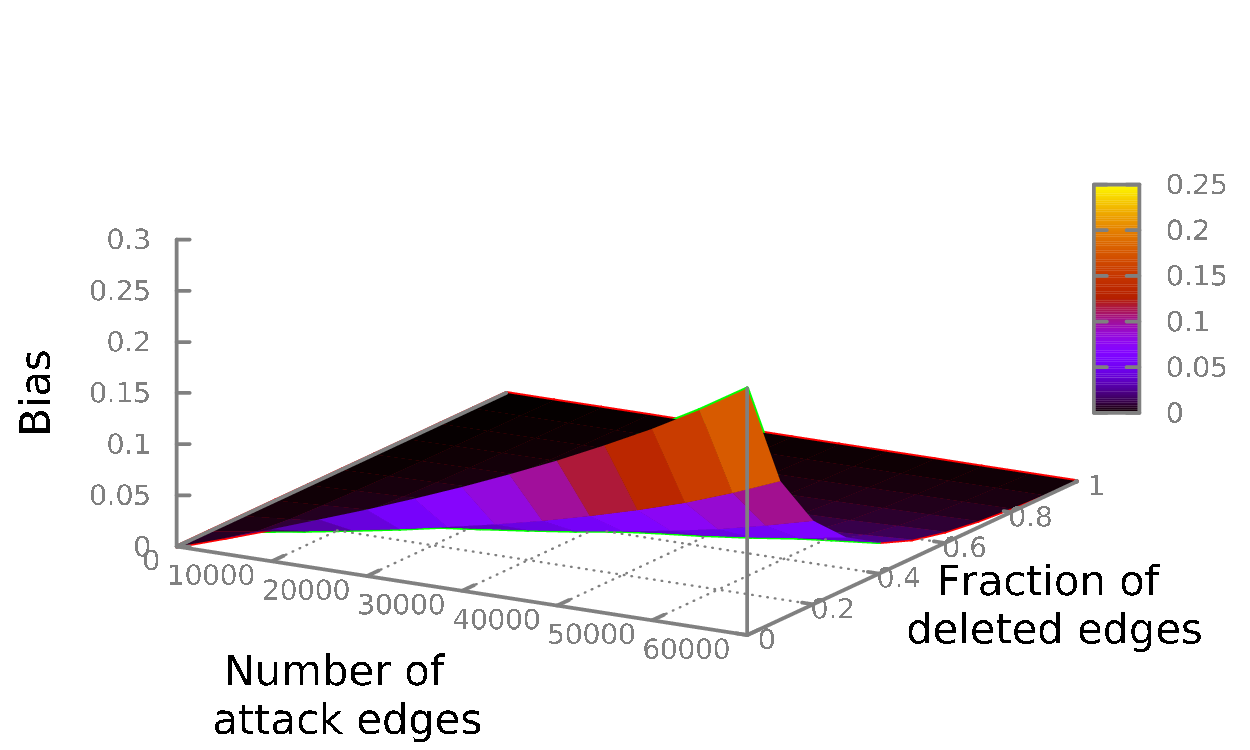,width=0.33\textwidth}\vspace{-0.00in}\\
{(c)}
\end{tabular}
}
\begin{ccs}
\vspace{-10pt}
\end{ccs}
\caption{{\em Probability of $l$'th hop being compromised (Sampling
    Bias) under route capture attack with global blacklisting [Facebook wall post graph]} (a)
  $l=1$, (b) $l=5$, (c) $l=25$ . As more edges to the honest nodes are
  removed, the attacker's loss is higher. %
}
\label{fig:global-blacklisting}
\begin{ccs}
\vspace{-10pt}
\end{ccs}
\end{figure*}

\subsection{Reciprocal Neighbor Policy}
To demonstrate the effectiveness of the reciprocal neighbor policy for implementing
trust-based anonymity, let us assume for now that there is a mechanism
to securely achieve the policy, i.e., that if a node $X$ does not advertise
a node $Y$ in its neighborlist, then $Y$ also excludes $X$. In this
scenario, we are interested in characterizing the probability
distribution of random walks. %
\begin{figure}[h]
\centering
\includegraphics[width=2.5in]{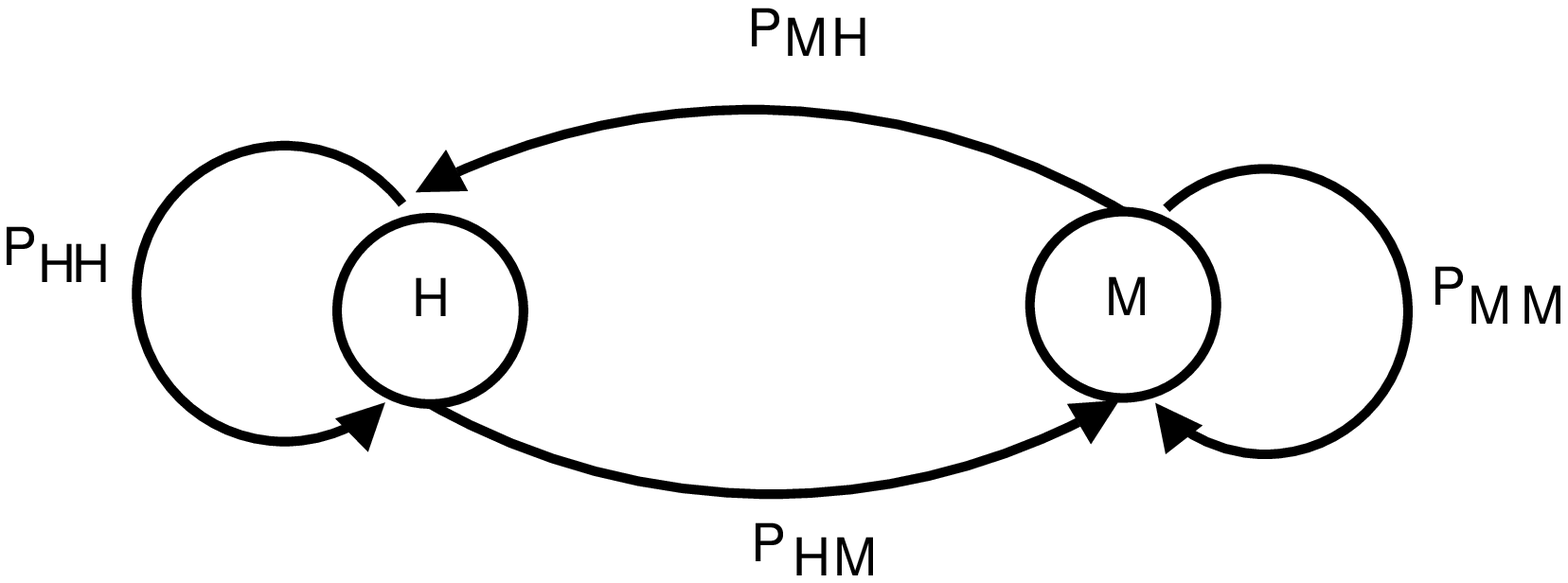}\\
\vspace{-10pt}
\caption{Attack Model.}
\label{fig:entanglement-state}
\vspace{-10pt}
\end{figure}

\begin{thm}\label{thm:node-degree}
{\em Node degree attack}: Given $h$ honest nodes and $m$ malicious nodes (including 
  Sybil nodes) that have $g$ edges (attack edges) amongst each other in an undirected and  connected social
  network, the stationary probability of random walks starting at an honest node
  and terminating at a malicious node cannot be biased by adding edges amongst 
  malicious nodes. Moreover, this stationary probability is independent of the 
  topology created amongst malicious nodes (as long as the social graph is connected).
\end{thm}
\begin{proof}
Let us denote $\pi_i$ as the stationary probability of random walks (independent of 
the initial state of the random walk) for node $i$, and let $P_{ij}$ denote the 
transition probability from node i to node j. Let $n$ denote the total number of 
nodes in the social network ($n=h+m$). Since the transition probabilities between 
nodes in the Metropolis-Hastings random walks are symmetric
($P_{ij}=P_{ji}=\min\left(\frac{1}{degree(i)},\frac{1}{degree(j)}\right)$),
observe that $\forall z, \pi_z=\frac{1}{n}$ is solution to the equation
$\pi_i \cdot P_{ij} = \pi_j \cdot P_{ji}$. Since social networks are
non-bipartite as well as undirected graphs, the solution to the above
equation ($\pi=\frac{1}{n}$) must be the unique stationary distribution
for the random walk~\cite{aldous-reversible}. Thus the stationary probability of 
random walks terminating at any node in the system is uniform and independent of 
the number of edges amongst malicious nodes, or the topology amongst malicious nodes 
in the system (as long as the graph remains connected). %
\end{proof}
We validate Theorem~\ref{thm:node-degree} using simulation results on
the Facebook wall post interaction
graph. Figure~\ref{fig:degree-attack}(a) depicts the probability of a
Pisces random walk terminating at a malicious node as a function of
random walk length for $g=30000$ (2900 malicious nodes) and $g=60000$
(7300 malicious nodes). We can see that the
random walk quickly reaches its stationary distribution, and at the
stationary distribution, the probability of a random walk terminating at
one of the malicious nodes is 0.1 and 0.25 respectively (which is the 
adversary's fair share).
Figure~\ref{fig:degree-attack}(b) and (c) depict the probability of a
random walk terminating at one of the malicious nodes under the node
degree attack, for $g=30000$ and $g=60000$ respectively. We can see that
adding edges amongst malicious nodes does not help the adversary (even for 
transient length random walks).
\begin{techreport}
\begin{proof}
To characterize the transient distribution of the random walk, we model
the process as a Markov chain (shown in
Figure~\ref{fig:entanglement-state}). Let us denote the honest set of
nodes by $H$, and the set of malicious nodes by $M$.

The probability of an $l$ hop random walk ending in the malicious region
($P(l)$) is given by:

\begin{align}
P(l) = P(l-1) \cdot P_{MM} + (1-P(l-1)) \cdot P_{HM}
\end{align}

The terminating condition for the recursion is $P(0)=1$, which reflects
that the initiator is honest.

We can estimate the probabilities $P_{HM}$ and $P_{MH}$ as the forward
and backward conductance~\cite{conductance} between the honest and the
malicious nodes, denoted by $\phi_F$ and $\phi_B$ respectively.  Thus we
have that:

\begin{align}
P(l) & =  P(l-1) \cdot (1-\phi_B) + (1-P(l-1)) \cdot \phi_F \nonumber \\
 & =  P(l-1) \cdot (1- \phi_B - \phi_F) + \phi_F
\label{eqn:secure-walks}
\end{align}

\begin{align}
P(l) =  \phi_F  \cdot [ 1 + (1- \phi_B - \phi_F) + (1-\phi_B-\phi_F)^2 \nonumber \\ 
 \ldots + (1-\phi_B-\phi_F)^{l-1} ] 
\label{eqn:secure-walks2}
\end{align}

With $g$ edges between honest and malicious nodes, we can estimate the forward conductance $\phi_F$ as follows: 

\begin{align}
\phi_F & = \frac{\Sigma_{x \in H} \Sigma_{y \in M} \pi_x \cdot P_{xy}}{\pi_H} \nonumber \\ 
&  = \frac{\Sigma_{x \in H} \Sigma_{y \in M} \cdot P_{xy}}{|H|}  = O\left(\frac{g}{h}\right)
\label{eqn:forward-conductance}
\end{align}

Similarly, with $g$ edges between honest and malicious nodes, the backward conductance $\phi_B$ is estimated as: 

\begin{align}
\phi_B & = \frac{\phi_F \cdot |H|} {|M|} = \frac{O(\frac{g}{h}) \cdot h} { m}  = O\left(\frac{g}{m}\right)   
\end{align}

Thus, we have that $\phi_F = O(\frac{g}{h})$, and
$\phi_B=O(\frac{g}{m})$. If malicious nodes exclude $y$ edges to honest
nodes from their fingertables, application of the RNP ensures that the
honest nodes also exclude the $y$ edges from their fingertables (local
blacklisting). %
Thus, route capture attacks result in deleting of attack edges which 
reduces both forward and backward transition probabilities. 

Observe that the probability of the first hop being in the malicious
region is equal to $\phi_F$, which gets reduced under attack. We will
now show this for a general value of $l$. Following
Equation~\ref{eqn:secure-walks2} and using $\Sigma_{i=0}^{i=m} x^i =
\frac{1-x^{m+1}}{1-x}$ for $0 < x < 1$, we have that:

\begin{align}
P(l) & = \frac{\phi_F  \cdot (1 - (1- \phi_B - \phi_F)^{l})}{1-(1-\phi_B-\phi_F)} \nonumber \\
& = \frac{\phi_F}{\phi_F+\phi_B} \cdot (1-(1-\phi_B+\phi_F)^l) 
\label{eqn:secure-walks3}
\end{align}

Using $\phi_B=\frac{h}{m} \cdot \phi_F$, we have that:

\begin{align}
P(l) & = \frac{m}{n} \cdot (1-(1-\phi_B+\phi_F)^l) \nonumber \\ 
P(l) & = \frac{m}{n} \cdot \left(1-\left(1-\frac{n}{m}\cdot \phi_F\right)^l\right) 
\label{eqn:secure-walks4}
\end{align}

Differentiating $P(l)$ with respect to $\phi_F$, we have that:

\begin{align}
\frac{d}{d\phi_F} (P(l)) & = \frac{m}{n} \cdot \left(l\cdot \left(1-\frac{n}{m} \cdot \phi_F\right)^{l-1}\right)
\end{align}

Note that $(1-\frac{n}{m}\phi_F) = (1-\phi_B-\phi_F) \geq 0$. This
implies $\frac{d}{d\phi_F}P(l) \geq 0$.  Thus, $P(l)$ is an increasing
function of $\phi_F$, and since the reduction of the number of attack
edges reduces $\phi_F$, it also leads to a reduction
in the transient distribution of the random walk terminating at
malicious nodes.
\end{proof}

Next, we validate our analysis using simulation.
Figure~\ref{fig:local-blacklisting} depicts the probability of a random
walk terminating at one of the malicious nodes as a function of the
number of attack edges as well as the fraction of sacrificed attack
edges. We can see that in all scenarios, this bias is a decreasing
function of the fraction of deleted edges. Moreover, it is notable that
the decrease is larger for shorter random walks than for longer random
walks. This reflects the fact that the stationary distribution under
local blacklisting being unchanged, which we address next.
\end{techreport}
\begin{thm}
{\em Global blacklisting}: suppose that $x \le m$ malicious nodes
sacrifice $y_1 \le g$ attack edges, and that %
these %
malicious nodes originally had $y_2 \le g$ attack edges. The stationary
probability of random walk terminating at malicious nodes gets reduced
proportional to $x$. The transient distribution of random walks
terminating at malicious nodes is reduced as a function of $y_2$.
\end{thm}
\begin{proof}
If $x$ malicious nodes perform the route capture attack and are globally
blacklisted, these nodes become disconnected from the social trust
graph.
It follows from our analysis of Theorem~\ref{thm:node-degree} that the
stationary distribution of the random walk is uniform for all
\emph{connected} nodes in the graph. Thus, the stationary distribution
of random walks terminating at malicious nodes gets reduced from
$\frac{m}{m+h}$ to $\frac{m-x}{m-x+h}$.

To characterize the transient distribution of the random walk, we model
the process as a Markov chain. %
 Let us denote the honest set of
nodes by $H$, and the set of malicious nodes by $M$.
The probability of an $l$ hop random walk ending in the malicious region
($P(l)$) is given by:

\begin{align}
P(l) = P(l-1) \cdot P_{MM} + (1-P(l-1)) \cdot P_{HM}
\end{align}

The terminating condition for the recursion is $P(0)=1$, which reflects
that the initiator is honest.
We can estimate the probabilities $P_{HM}$ and $P_{MH}$ as the forward
and backward conductance~\cite{conductance} between the honest and the
malicious nodes, denoted by $\phi_F$ and $\phi_B$ respectively.  Thus we
have that:

\begin{align}
P(l) & =  P(l-1) \cdot (1-\phi_B) + (1-P(l-1)) \cdot \phi_F \nonumber \\
 & =  P(l-1) \cdot (1- \phi_B - \phi_F) + \phi_F
\label{eqn:secure-walks}
\end{align}

\begin{align}
P(l) =  \phi_F  \cdot [ 1 + (1- \phi_B - \phi_F) + (1-\phi_B-\phi_F)^2 \nonumber \\
 \ldots + (1-\phi_B-\phi_F)^{l-1} ]
\label{eqn:secure-walks2}
\end{align}

We note that if an adversary connects a chain of Sybils (say of degree 2) to an attack edge, a random 
walk starting from an honest node and traversing the attack edge to enter the malicious region has a 
non-trivial probability of coming back to the honest region - via the attack edge (Pisces allows backward
transition along edges). Our analysis models the probability of returning to the honest region using 
the notion of backward conductance.  

With $g$ edges between honest and malicious nodes, we can estimate the forward conductance $\phi_F$ as follows:

\begin{align}
\phi_F & = \frac{\Sigma_{x \in H} \Sigma_{y \in M} \pi_x \cdot P_{xy}}{\pi_H} \nonumber \\
&  = \frac{\Sigma_{x \in H} \Sigma_{y \in M} \cdot P_{xy}}{|H|}  = O\left(\frac{g}{h}\right)
\label{eqn:forward-conductance}
\end{align}

Similarly, with $g$ edges between honest and malicious nodes, the backward conductance $\phi_B$ is estimated as:

\begin{align}
\phi_B & = \frac{\phi_F \cdot |H|} {|M|} = \frac{O(\frac{g}{h}) \cdot h} { m}  = O\left(\frac{g}{m}\right)
\end{align}

Thus, we have that $\phi_F = O(\frac{g}{h})$, and
$\phi_B=O(\frac{g}{m})$. If malicious nodes exclude $y$ edges to honest
nodes from their fingertables, application of the RNP ensures that the
honest nodes also exclude the $y$ edges from their fingertables (local
blacklisting). %
Thus, route capture attacks result in deleting of attack edges which
reduces both forward and backward transition probabilities.
Observe that the probability of the first hop being in the malicious
region is equal to $\phi_F$, which gets reduced under attack. We will
now show this for a general value of $l$. Following
Equation~\ref{eqn:secure-walks2} and using $\Sigma_{i=0}^{i=m} x^i =
\frac{1-x^{m+1}}{1-x}$ for $0 < x < 1$, we have that:

\begin{align}
P(l) & = \frac{\phi_F  \cdot (1 - (1- \phi_B - \phi_F)^{l})}{1-(1-\phi_B-\phi_F)} \nonumber \\
& = \frac{\phi_F}{\phi_F+\phi_B} \cdot (1-(1-\phi_B+\phi_F)^l)
\label{eqn:secure-walks3}
\end{align}

Using $\phi_B=\frac{h}{m} \cdot \phi_F$, we have that:

\begin{align}
P(l) & = \frac{m}{n} \cdot (1-(1-\phi_B+\phi_F)^l) \nonumber \\
P(l) & = \frac{m}{n} \cdot \left(1-\left(1-\frac{n}{m}\cdot \phi_F\right)^l\right)
\label{eqn:secure-walks4}
\end{align}

Differentiating $P(l)$ with respect to $\phi_F$, we have that:

\begin{align}
\frac{d}{d\phi_F} (P(l)) & = \frac{m}{n} \cdot \left(l\cdot \left(1-\frac{n}{m} \cdot \phi_F\right)^{l-1}\right)
\end{align}

Note that $(1-\frac{n}{m}\phi_F) = (1-\phi_B-\phi_F) \geq 0$. This
implies $\frac{d}{d\phi_F}P(l) \geq 0$.  Thus, $P(l)$ is an increasing
function of $\phi_F$, and since the reduction of the number of attack
edges reduces $\phi_F$, it also leads to a reduction
in the transient distribution of the random walk terminating at
malicious nodes. Thus, it follows that a reduction in the number
of remaining attack edges $y_2$ reduces the transient distribution
of random walks terminating at malicious nodes.
\end{proof}

Figure~\ref{fig:global-blacklisting} depicts the probability of random
walks terminating at malicious nodes as a function of number of attack
edges as well as the fraction of deleted edges when honest nodes use a
global blacklisting policy. We can see that sacrificing attack
edges so as to perform route capture attacks is a losing strategy for
the attacker. Moreover, the decrease is similar for all random walk
lengths; this is because even the stationary distribution of the random
walk terminating at malicious nodes is reduced.

{\bf Anonymity Implication:} To de-anonymize the user without the help
of the destination node (e.g. the website to which the user connects
anonymously), both the first hop and the last hop of the random walk
need to be malicious to observe the connecting user and her
destinations, respectively. End-to-end timing analysis~\cite{wang:ccs05,
  rainbow} makes it so that controlling these two nodes is sufficient
for de-anonymization. Figure~\ref{fig:e2e-global} depicts the
probability of such an attack being successful as a function of the
number of attack edges and the fraction of deleted edges using the
global blacklisting policy. %
We can see that the probability of attack is %
a decreasing function of the fraction of deleted edges.%
Thus we conclude that route capture attacks are %
a losing strategy %
against our approach.

So far, we validated our analysis using simulations assuming an %
an ideal Sybil defense. We also validated our analysis %
using a more realistic Sybil defense that permits a bounded number (set to 10~\cite{sybillimit}) 
of Sybils per attack edge, which we show in Figure~\ref{fig:e2e-global-sybil}.%
\begin{techreport}
\begin{figure}[p]
\centering
\includegraphics[width=2.5in]{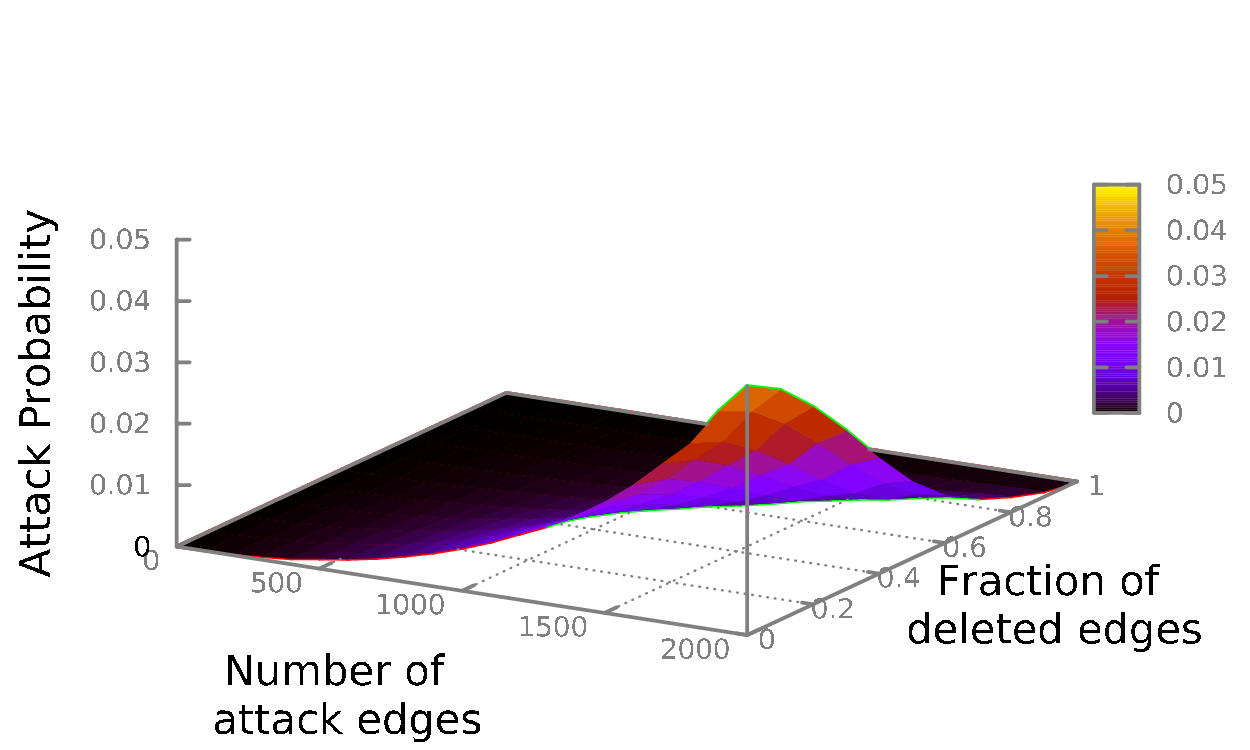}\\
\caption{Probability of end-to-end timing analysis under route capture
  attack}
\label{fig:e2e-local}
\end{figure}
\end{techreport}

\begin{figure}[tp]
\centering
\includegraphics[width=2.5in]{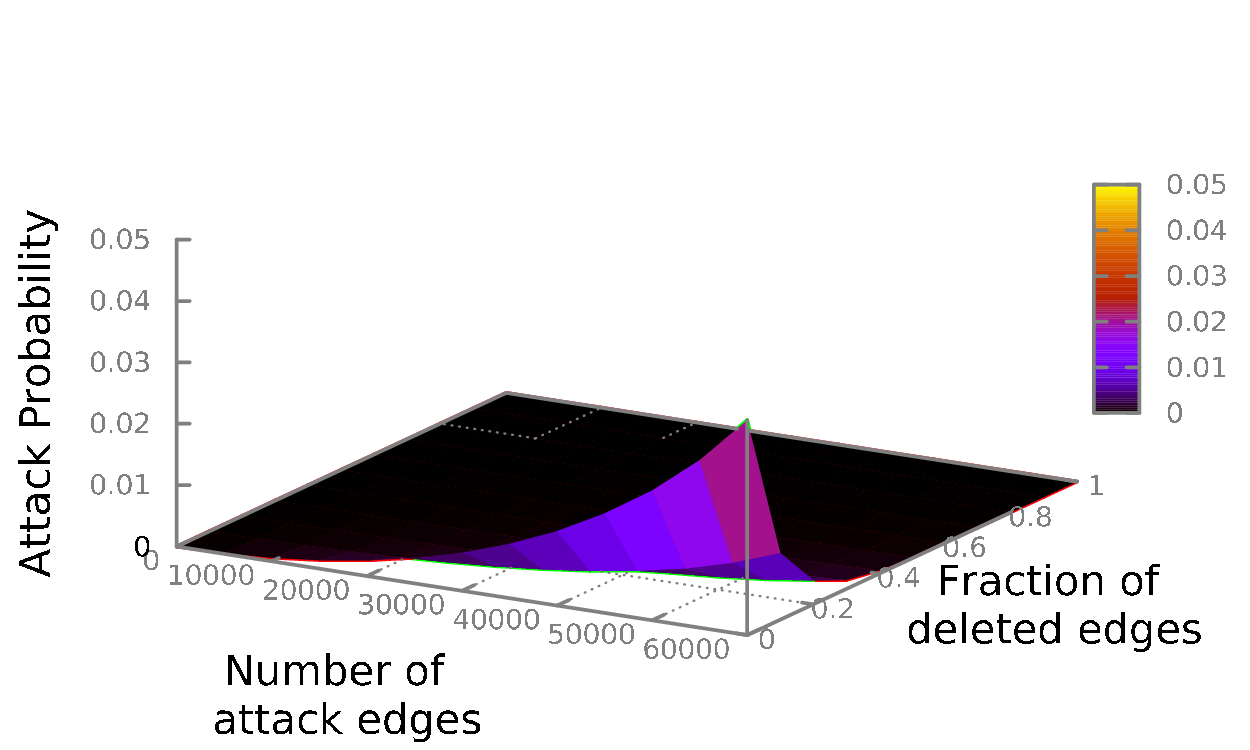}\\
\begin{ccs}
\vspace{-10pt}
\end{ccs}
\caption{Probability of end-to-end timing analysis under route capture
  attack with global blacklisting [Facebook wall graph]}
\label{fig:e2e-global}
\begin{ccs}
\vspace{-10pt}
\end{ccs}
\end{figure}

\begin{figure}[tp]
\centering
\includegraphics[width=2.5in]{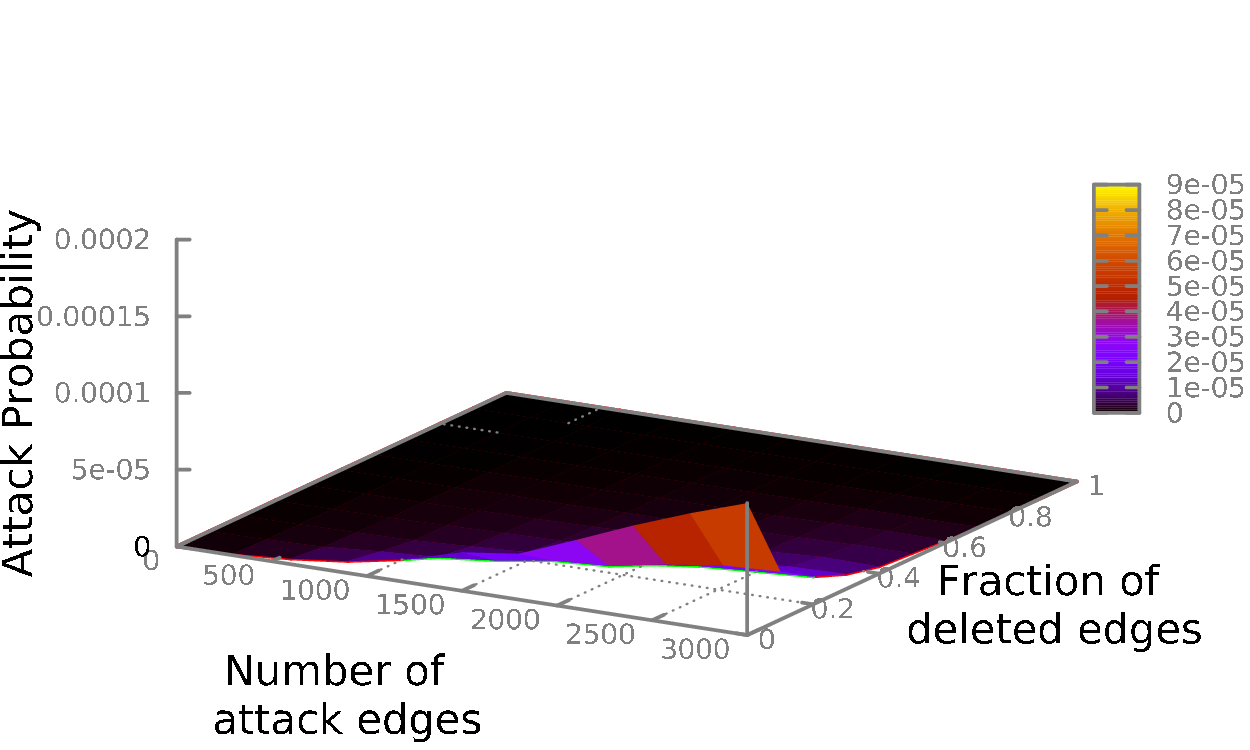}\\
\begin{ccs}
\vspace{-10pt}
\end{ccs}
\caption{Probability of end-to-end timing analysis under route capture
  attack with global blacklisting using 10 Sybils per attack edge [Facebook wall graph]}
\label{fig:e2e-global-sybil}
\begin{ccs}
\vspace{-10pt}
\end{ccs}
\end{figure}

\subsection{Securing Reciprocal Neighborhood Policy}
We now discuss the security and performance of our protocol that
implements the reciprocal neighbor policy.

{\bf Security proof sketch:} Suppose that a malicious node $A$
aims to exclude an honest node $B$ from its neighborlist. To pass node
$B$'s local integrity checks, node $A$ has to return a neighborlist to
node $B$ that correctly advertises node $B$. Since random walks for
anonymous communication are indistinguishable from testing random walks,
there is a probability that the adversary will advertise a conflicting
neighbor list that does not include node $B$ to an initiator of the
testing random walk. The initiator of the testing random walk will
insert the malicious neighbor list into the Whanau DHT, and node $B$ can
perform a robust lookup for node $A$'s key to obtain the conflicting
neighbor list. Since Whanau only provides availability, node $B$ can
check for integrity of the results by verifying node $A$'s
signature. Since honest nodes never advertise two conflicting lists
within a time interval, node $B$ can infer that node $A$ is malicious.
{\bf Performance Evaluation:}
\label{sec:perf}
We analyze the number of testing random walks that each node must
perform to achieve a high probability of detecting a malicious node that
attempts to perform a route capture attack. Nodes must perform enough
testing walks such that a high percentage of compromised nodes (which
are connected to honest nodes) have been probed in a single time
slot. First, we consider a defense strategy where honest nodes only
insert the terminal hop of the testing random walks in Whanau ({\em
  Strategy 1}). Intuitively, from the coupon collectors problem, $\log
n$ walks per node should suffice to catch a malicious node with high
probability. Indeed, from Figure~\ref{fig:detection-prob}, we can see
that six testing walks per time interval suffice to catch a malicious
node performing route capture attacks with high probability. The honest
nodes can also utilize all hops of the testing random walks to check for
conflicts ({\em Strategy 2}), in which case only two or three testing
walks are required per time interval (at the cost of increased
communication overhead for the DHT operations).
\begin{figure}[tp]
\centering
\includegraphics[width=2.5in]{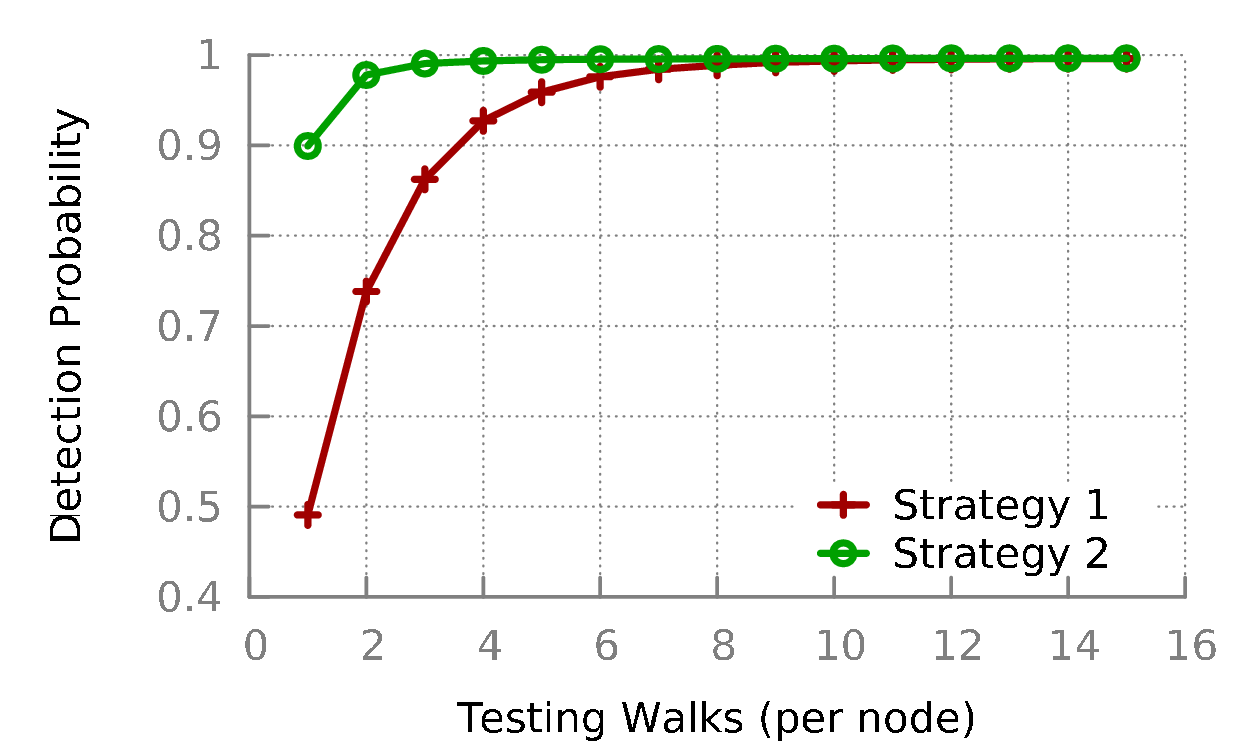}\\
\begin{ccs}
\vspace{-10pt}
\end{ccs}
\caption{Probability of detecting a route capture [Facebook wall post
    interaction graph]. The attack model includes 10 Sybils per attack
  edge.}
\label{fig:detection-prob}
\begin{ccs}
\vspace{-10pt}
\end{ccs}
\end{figure}

Next, we address the question of how to choose the duration of the time
interval ($t$). The duration of the time slot governs the trade-off
between communication overhead and reliability of circuit
construction. A large value of the time slot interval results in a
smaller communication overhead but higher unreliability in circuit
construction, since nodes selected in the random walk are less likely to
still be online. On the other hand, a smaller value of the time interval
provides higher reliability in higher circuit construction at the cost
of increased communication overhead, since a fixed number of testing
walks must be performed within the duration of the time slot. We can see
this trade-off in Figure~\ref{fig:unreliability}. We consider two churn
models for our analysis: (a) nodes have a mean lifetime of 24 hours
(reflecting behavior of Tor relays~\cite{torstatus-blutmagie}
), and (b) nodes have a mean lifetime of 1 hour (reflecting
behavior of conventional P2P networks). For the two scenarios, using a
time slot duration of 3 hours and one of 5 minutes, respectively,
results in a 2-3\% probability of getting an unreliable random walk for
up to 25 hops.
\begin{figure}[tp]
\centering
\includegraphics[width=2.5in]{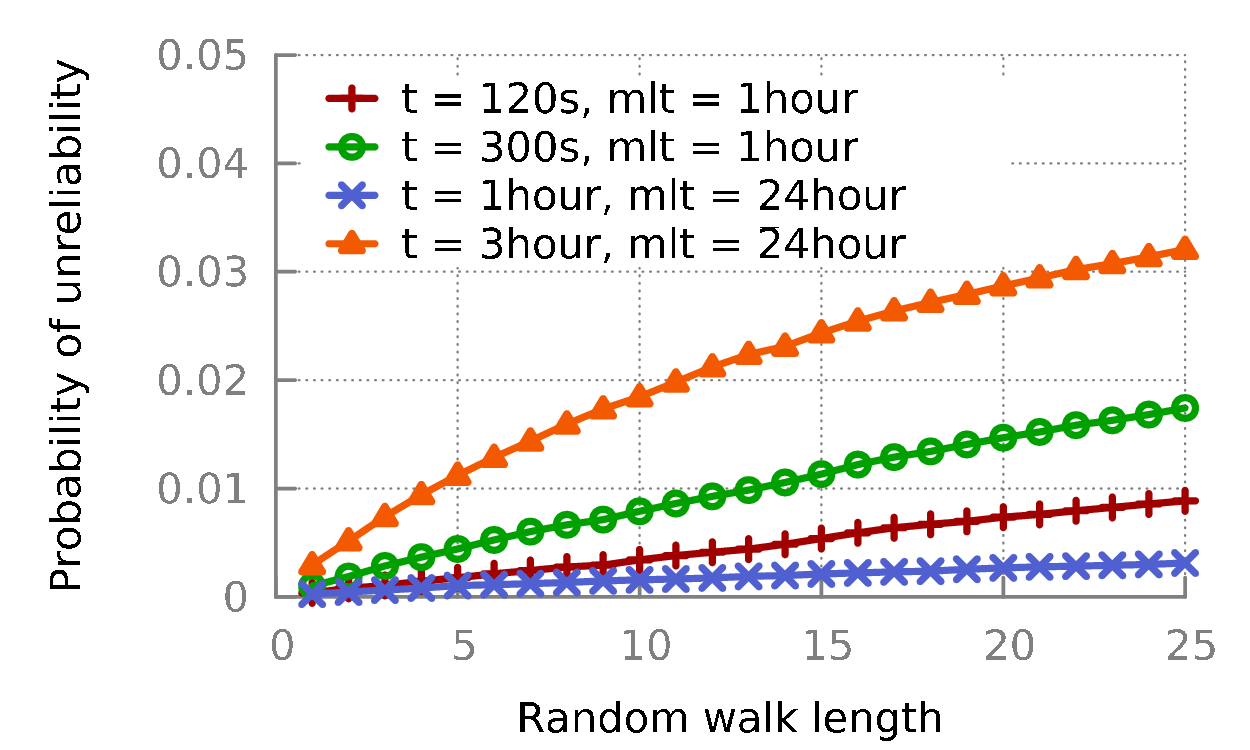}\\
\begin{ccs}
\vspace{-10pt}
\end{ccs}
\caption{Unreliability in circuit construction [Facebook wall post
    interaction graph].}
\label{fig:unreliability}
\begin{ccs}
\vspace{-10pt}
\end{ccs}
\end{figure}

{\bf Overhead:}
There are three main sources of communication overhead in our
system. First is the overhead due to setting up the neighbor lists;
which requires about $d^2$\ KB of communication, where $d$ is the node
degree.
The second source of overhead is the testing random walks, where nodes
are required to perform about six such walks of length 25. %
The third source of overhead comes from participation in the Whanau
DHT. Typically, key churn is a significant source of overhead in Whanau,
requiring all of its routing tables to be rebuilt. However, in our
scenario, only the values corresponding to the keys change quickly, but
not the keys themselves, requiring only a modest amount of heartbeat
traffic~\cite{whanau:nsdi10}. Considering the Facebook wall post
topology, we estimate the mean communication overhead per node \emph{per 
time interval} to be only about 6\ MB.%
\begin{figure}[tp]
\centering
\includegraphics[width=2.5in]{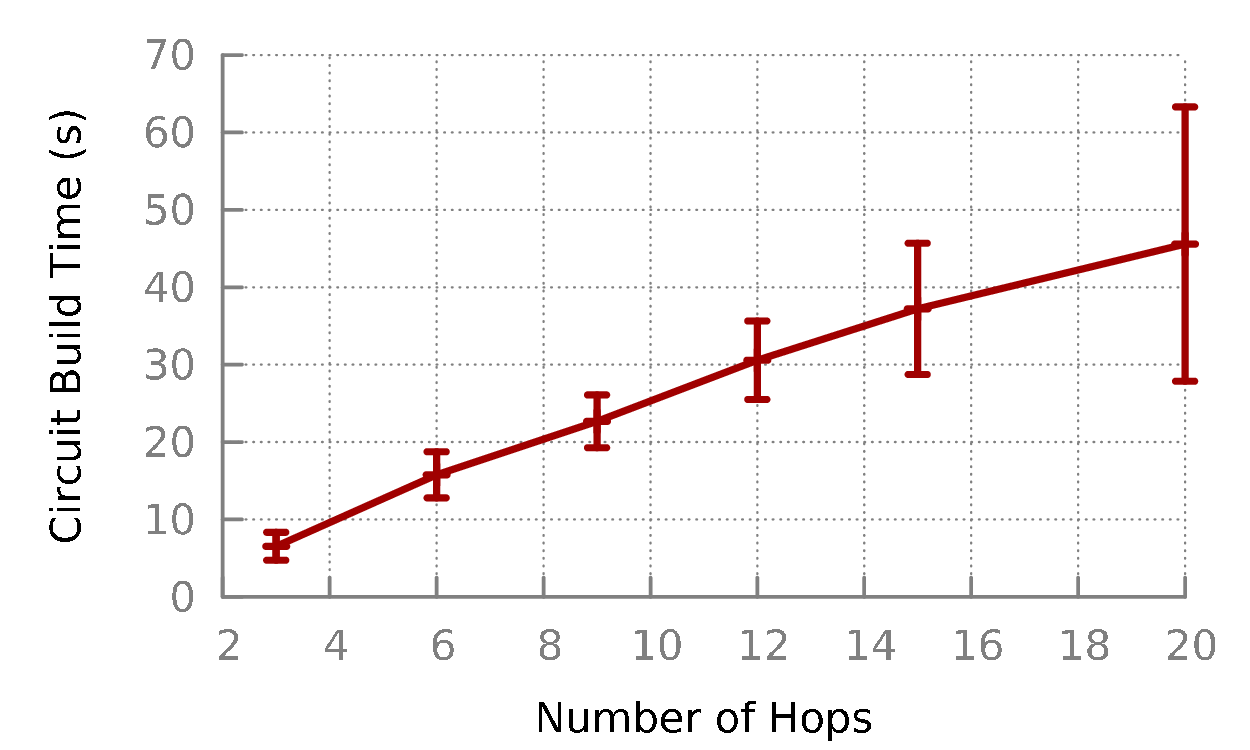}\\
\begin{ccs}
\vspace{-10pt}
\end{ccs}
\caption{Circuit build times in Tor as a function of circuit length}
\label{fig:build-times}
\begin{ccs}
\vspace{-10pt}
\end{ccs}
\end{figure}
We also evaluate the latency of constructing long onion routing circuits
through experiments over the live Tor network. We used the Torflow
utility to build Tor circuits with varying circuit lengths;
Figure~\ref{fig:build-times} depicts our experimental results. Using
these results, we estimate that 25 hop circuits would take about 1
minute to construct.

\subsection{Anonymity}

Earlier, we considered the probability of end-to-end timing analysis as
our metric for anonymity. This metric considers the scenario where the
adversary has exactly de-anonymized the user. However, in random walk
based anonymous communication, the adversary may sometimes have
probabilistic knowledge of the initiator. To quantify all possible
sources of information leaks, we now use the entropy metric to quantify
anonymity~\cite{diaz:pets02,serjantov:pets02}. The entropy metric
considers the \emph{probability distribution} of nodes being possible
initiators, as computed by the attackers. In this paper, we will
restrict our analysis to Shannon entropy, since it is the most widely
used mechanism for analyzing anonymity. There are other ways of
computing entropy, such as guessing entropy~\cite{guessing-entropy} and
min entropy~\cite{min-entropy}, which we will consider in the full
version of this work. Shannon entropy is computed as:
\begin{align}
H(I) & = \sum_{i=0}^{i=n} -p_i \cdot \log_2(p_i)
\end{align}
where $p_i$ is the probability assigned to node $i$ of being the
initiator of a circuit. Given a particular observation $o$, the
adversary can first compute the probability distribution of nodes being
potential initiators of circuits $p_i | o$ and then the corresponding
conditional entropy $H(I|o)$. We can model the entropy distribution of
the system as a whole by considering the weighted average of entropy for
each possible observation, including the null observation.
\begin{align}
H(I|O) = \sum_{o \in O} P(o) \cdot H(I|o)
\label{eqn:system-entropy}
\end{align}
We first consider the scenario where an initiator performs an
$l$-hop random walk to communicate with a malicious destination, and the
nodes in the random walk are all honest, i.e., the adversary is external
to the system. For this scenario, we will analyze the expected initiator
anonymity under the conservative assumption that the adversary has
complete knowledge of the entire social network graph.

{\bf Malicious destination:}
\begin{figure*}[ht]
\centering
\mbox{
\hspace{-0.2in}
\hspace{-0.12in}
\begin{tabular}{c}
\psfig{figure=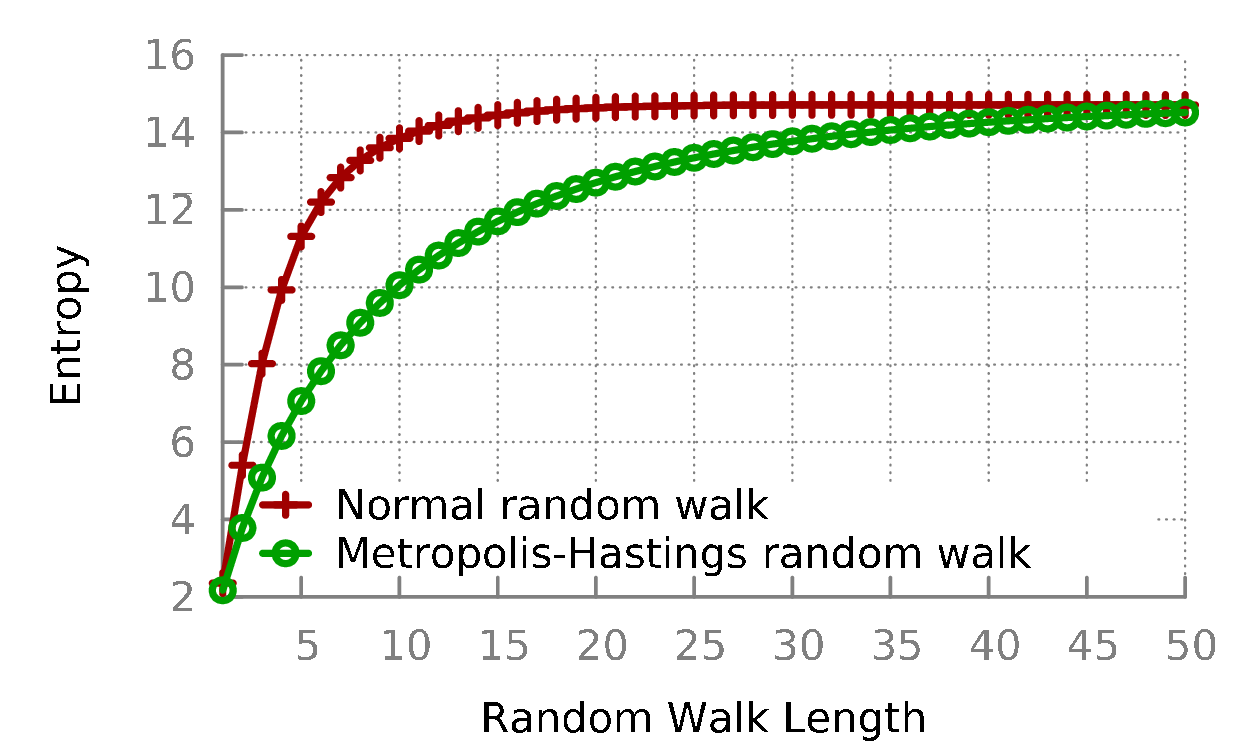,width=0.33 \textwidth}\vspace{-0.00in}\\
{(a)}
\end{tabular}
\begin{tabular}{c}
\psfig{figure=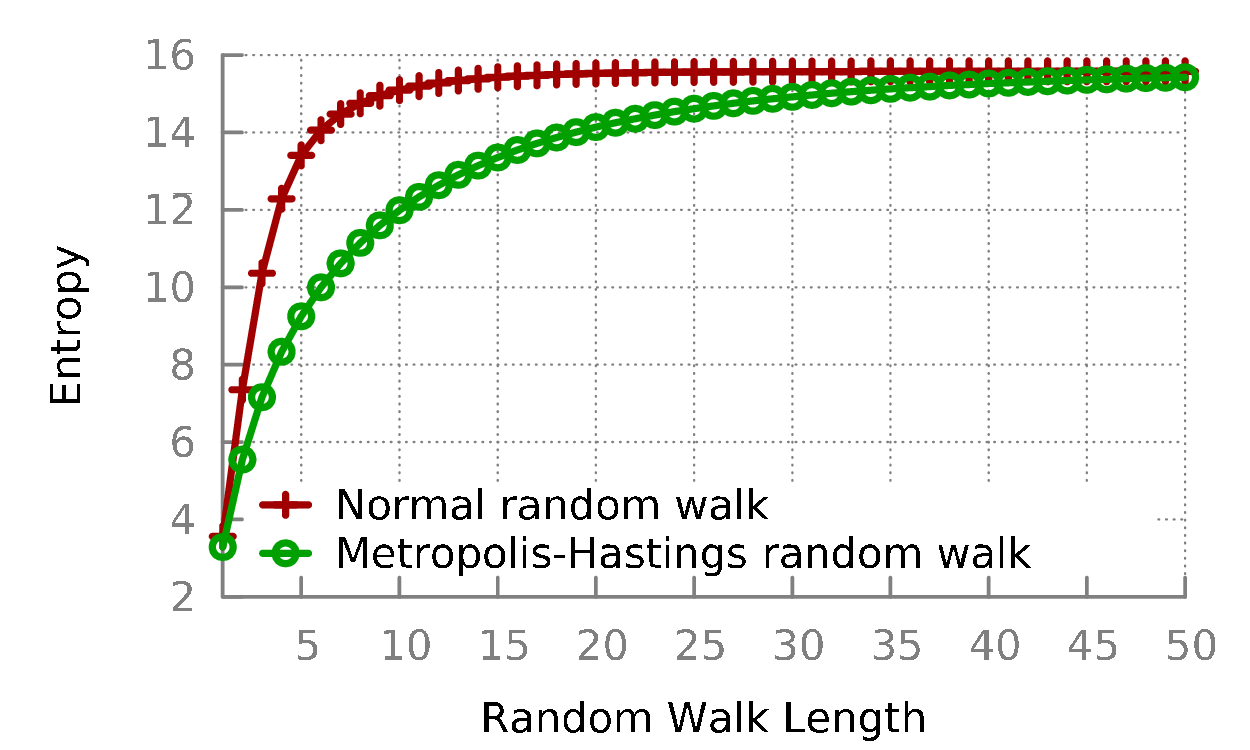,width=0.33 \textwidth}\vspace{-0.00in}\\
{(b)}
\end{tabular}
\hspace{-0.2in}
\begin{tabular}{c}
\psfig{figure=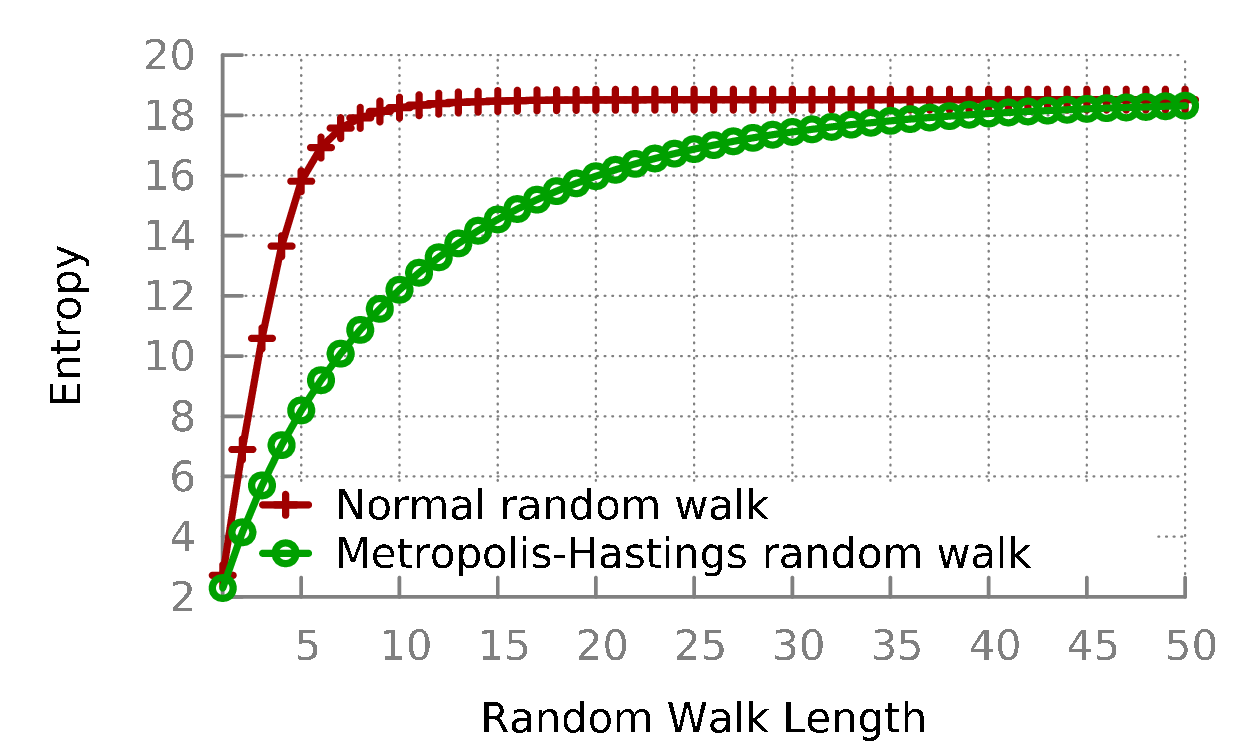,width=0.33\textwidth}\vspace{-0.00in}\\
{(c)}
\end{tabular}
}
\mbox{
\hspace{-0.2in}
\hspace{-0.12in}
\begin{tabular}{c}
\psfig{figure=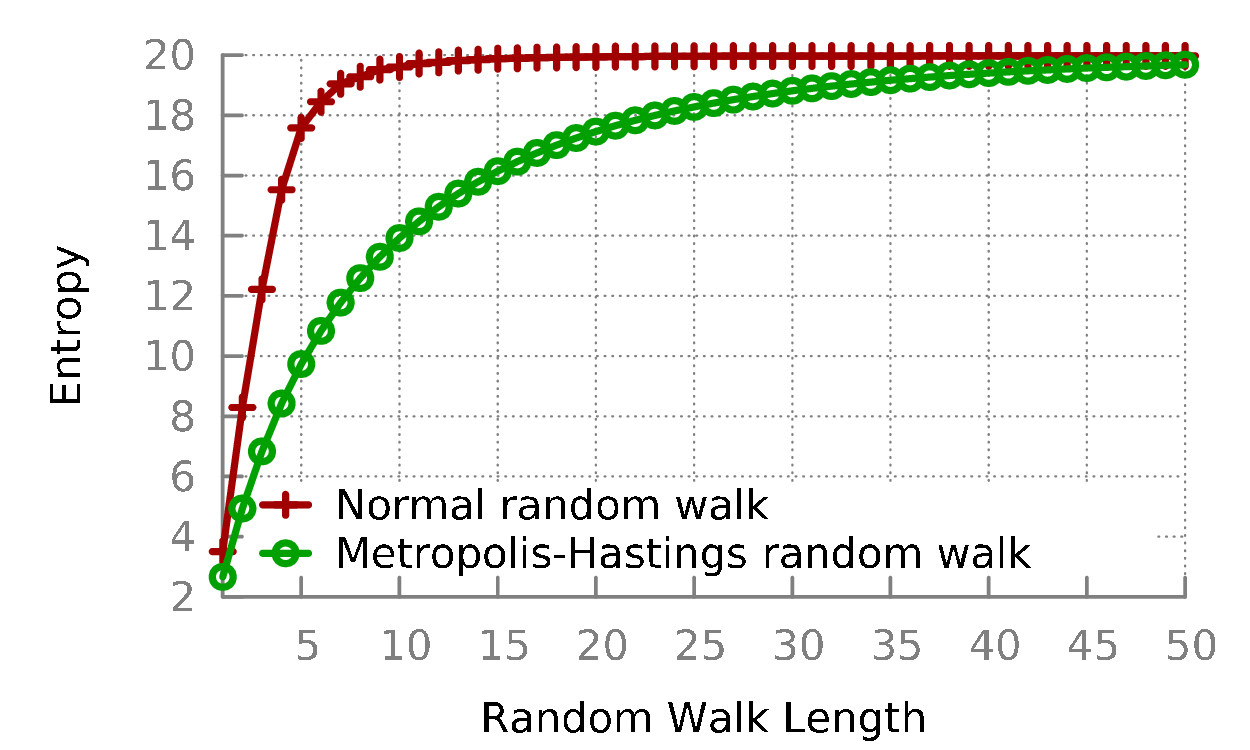,width=0.33 \textwidth}\vspace{-0.00in}\\
{(d)}
\end{tabular}
\begin{tabular}{c}
\psfig{figure=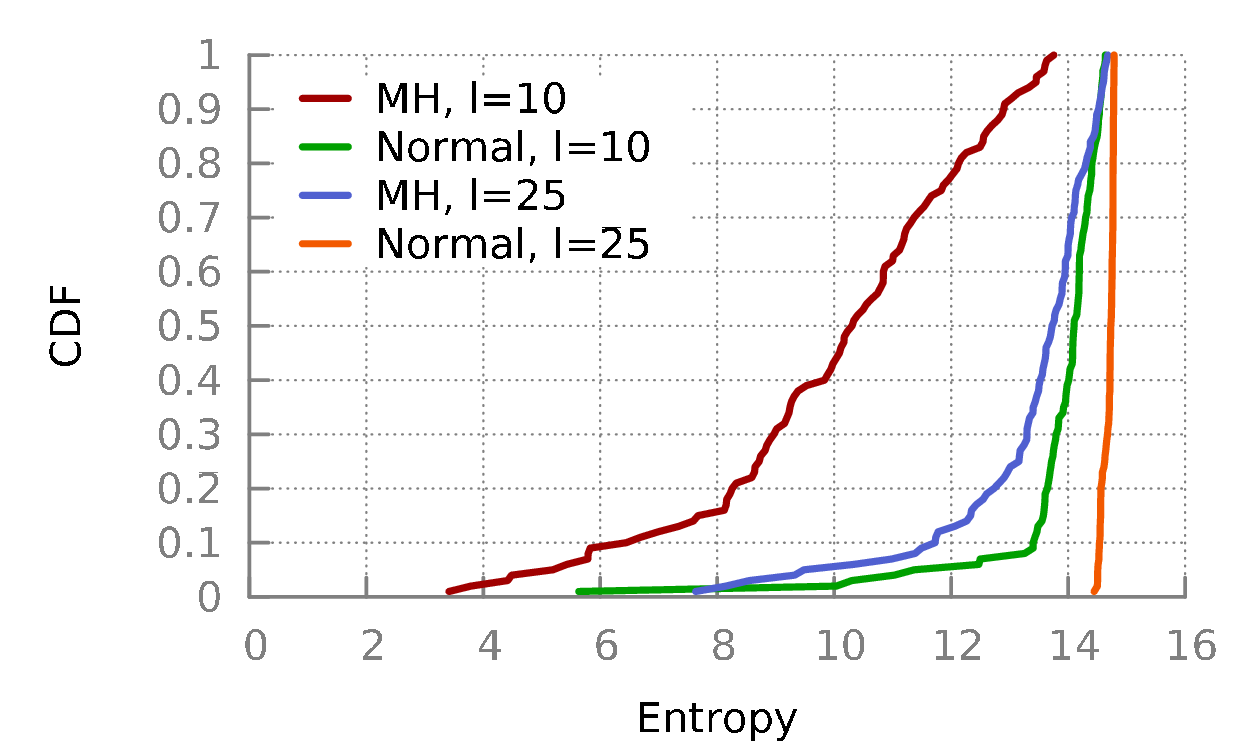,width=0.33 \textwidth}\vspace{-0.00in}\\
{(e)}
\end{tabular}
}
\begin{ccs}
\vspace{-5pt}
\end{ccs}
\caption{{\em Expected entropy as a function of random walk length using
    (a) Facebook wall post graph (b) Facebook link graph (c) Anonymous
    Interaction graph (d) Anonymous link graph and (e) CDF of entropy
    for Facebook wall graph and malicious destination.}  %
  }
\label{fig:entropy}
\begin{ccs}
\vspace{-10pt}
\end{ccs}
\end{figure*}
A naive way to compute initiator entropy for this scenario is to
consider the set of nodes that are reachable from the terminus of the
random walk in exactly $l$ hops (the adversary's observation), and
assign a uniform probability to all nodes in that set of being potential
initiators. However, such an approach does not consider the mixing
characteristics of the random walk; $l$ hop random walks starting at
different nodes may in fact have heterogeneous probabilities of
terminating at a given node.

We now outline a strategy that explicitly considers the mixing
characteristics of the trust topology. Let the terminus of an $l$ hop
random walk be node $j$. The goal of the adversary is to compute
probabilities $p_i$ of a particular node being the initiator of the
circuit:
\begin{align}
p_i & = \frac{P^l_{ij}}{\sum_x P^l_{xj}}
\label{eqn:initiator-prob}
\end{align}
where $P^l$ denotes the $l$-hop transition matrix for the random walk
process. Note that even for moderate size social graphs, the explicit
computation of $P^l$ is infeasible in terms of both memory and
computational constraints. This is because even though $P$ is a sparse
matrix, iterative multiplication of $P$ by itself results in a matrix
that is no longer sparse. To make the problem feasible, we propose to
leverage the \emph{time reversibility} of our random walk process. We
have previously modeled the random walk process as a Markov
chain (Theorem 2 in Appendix). Markov chains that satisfy the following property are known as
time-reversible Markov chains~\cite{aldous-reversible}.
\begin{align}
\pi_i \cdot P_{ij} = \pi_j \cdot P_{ji}
\end{align}
Both the conventional random walk and the Metropolis-Hastings random
walk satisfy the above property and are thus time-reversible Markov
chains. It follows from time reversibility~\cite{aldous-reversible},
that:
\begin{align}
\pi_i \cdot P^l_{ij} = \pi_j \cdot P^l_{ji} \implies P^l_{ij} = \frac{\pi_j}{\pi_i} \cdot P^l_{ji} 
\end{align}
Thus it is possible to compute $P^l_{ij}$ using $P^l_{ji}$. Let $V_j$ be
the initial probability vector starting at node $j$. Then the
probability of an $l$ hop random walk starting at node $j$ and ending at
node $i$ can be computed as the i'th element of the vector $V_j \cdot
P^l$. Observe that $V_j \cdot P^l$ can be computed without computing the
matrix $P^l$:
\begin{align}
V_j \cdot P^ l & = (V_j \cdot P) \cdot P^{l-1}
\end{align}
Since $P$ is a sparse matrix, $V_j \cdot P$ can be computed in $O(n)$
time, and $V_j \cdot P^l$ can be computed in $O(nl)$ time. Finally, we
can compute the probabilities of nodes being potential initiators of
circuits using equation~\eqref{eqn:initiator-prob}, and the
corresponding entropy gives us the initiator anonymity. 
\prateekccs{We average the
resulting entropy over 100 randomly chosen %
terminal nodes $j$ to compute the expected anonymity.}

Figure~\ref{fig:entropy}(a)-(d) depicts the expected initiator anonymity
as a function of random walk length for different social network
topologies. We can see that longer random walks result in an increase in
anonymity. This is because for short random walks of length $l$ in
restricted topologies such as trust networks, not every node can reach
the terminus of the random walk in $l$ hops. Secondly, even for nodes
that can reach the terminus of the random walk in $l$ hops, the
probabilities of such a scenario happening can be highly
heterogeneous. Further more, we can see that conventional random walks
converge to optimal entropy in about 10 hops for all four topologies. In
contrast, the Metropolis-Hastings random walks used in Pisces take
longer to converge. This is because random walks in Pisces have slower
mixing properties than conventional random walks. However, we can see
that even the Metropolis-Hastings random walk starts to converge after
25 hops in all scenarios.  

To get an understanding of the distribution of the entropy, we plot the
CDF of entropy over 100 random walk samples in
Figure~\ref{fig:entropy}(e). We can see that the typical anonymity
offered by moderately long random walks is %
high. For example,
using a Metropolis-Hastings random walk of length $25$, 95\% of users
get an entropy of at least $11$ bits.
So far, we observed %
that Metropolis-Hastings random walks need to be longer than
conventional random walks for equivalent level of anonymity against a
malicious destination. Next, we will see the benefit of using
Metropolis-Hastings random walks in Pisces, \prateekccs{since %
they can be
secured against insider attacks.%
}

{\bf Insider adversary:}
\begin{figure*}[htp]
\centering
\mbox{
\hspace{-0.2in}
\hspace{-0.12in}
\begin{tabular}{c}
\psfig{figure=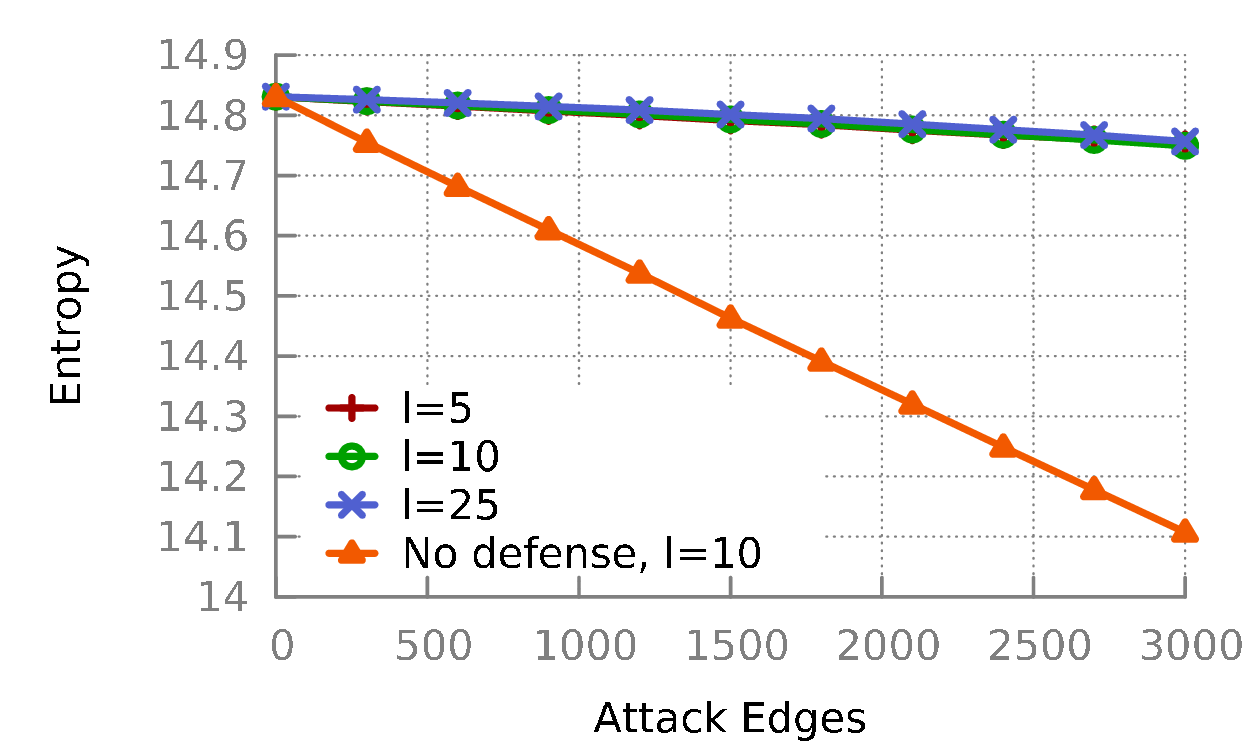,width=0.33 \textwidth}\vspace{-0.00in}\\
{(a)}
\end{tabular}
\begin{tabular}{c}
\psfig{figure=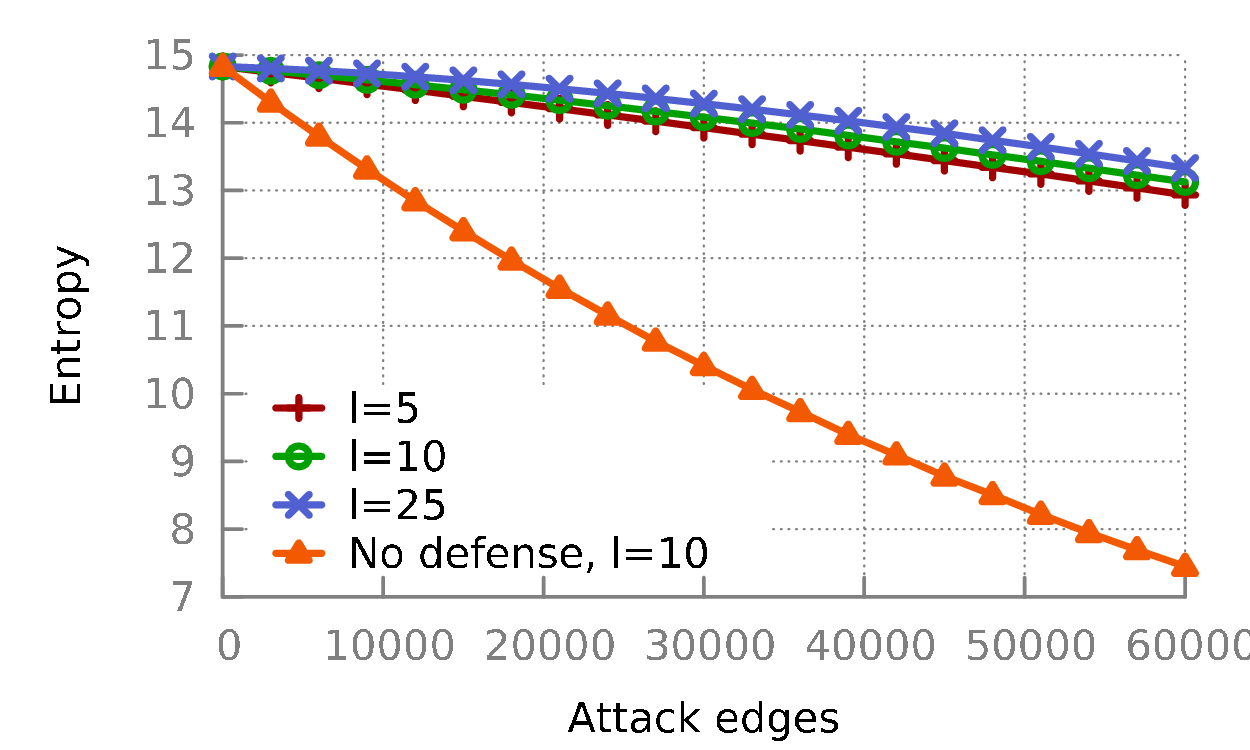,width=0.33 \textwidth}\vspace{-0.00in}\\
{(b)}
\end{tabular}
}
\begin{ccs}
\vspace{-5pt}
\end{ccs}
\caption{{\em Entropy as a function of fraction of attack edges using}
  (a) realistic model of an imperfect Sybil defense (10 Sybils per
  attack edge) and (b) perfect Sybil defense for Facebook wall post
  interaction graph. Note that the "No Defense" strategy models the random walks 
 used in systems such as Drac and Whanau. }
\label{fig:entropy-malicious}
\begin{ccs}
\vspace{-10pt}
\end{ccs}
\end{figure*}
\begin{ccs}
\begin{figure*}[htp]
\centering
\mbox{
\hspace{-0.2in}
\hspace{-0.12in}
\begin{tabular}{c}
\psfig{figure=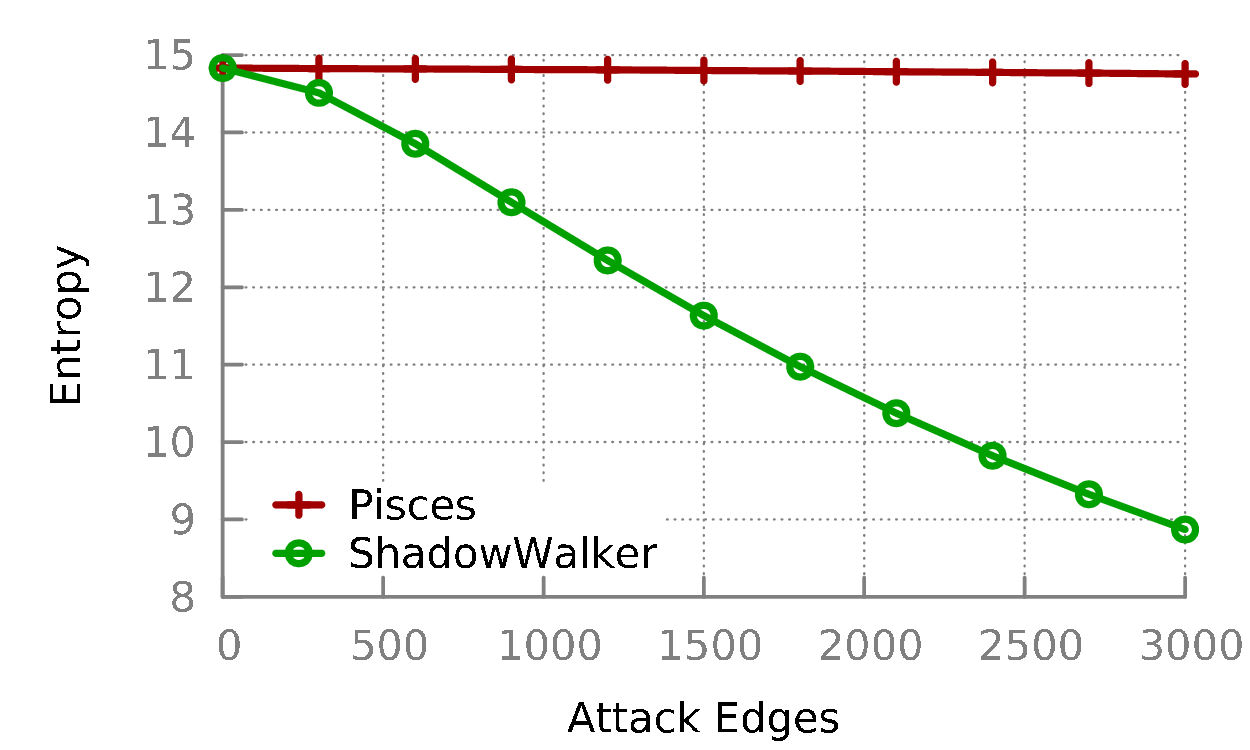,width=0.33 \textwidth}\vspace{-0.00in}\\
{(a)}
\end{tabular}
\begin{tabular}{c}
\psfig{figure=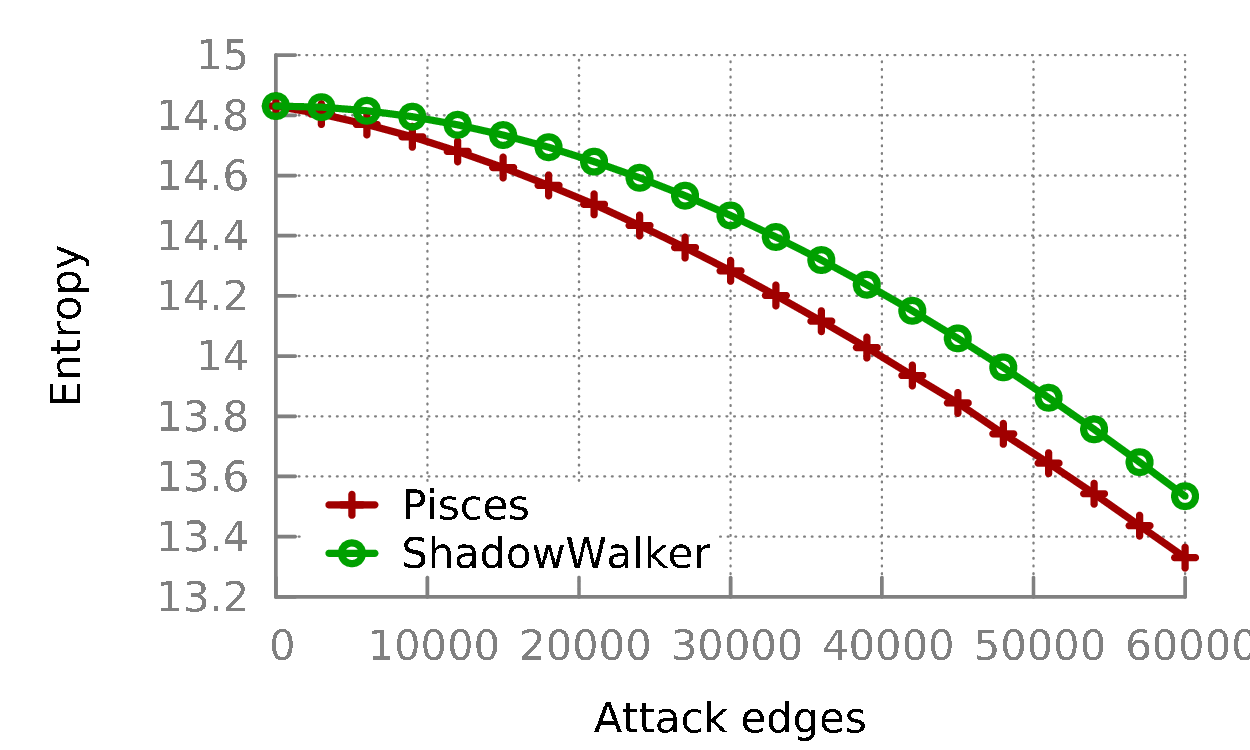,width=0.33 \textwidth}\vspace{-0.00in}\\
{(b)}
\end{tabular}
}
\begin{ccs}
\vspace{-5pt}
\end{ccs}
\caption{{\em Comparison with ShadowWalker. Entropy as a function of
    fraction of attack edges using} (a) realistic model of an imperfect
  Sybil defense (10 Sybils per attack edge) and (b) perfect Sybil
  defense for Facebook wall post interaction graph.  }
\label{fig:entropy-shadowwalker}
\vspace{-10pt}
\end{figure*}
\end{ccs}
\begin{techreport}
\begin{figure*}[thp]
\centering
\mbox{
\hspace{-0.2in}
\hspace{-0.12in}
\begin{tabular}{c}
\psfig{figure=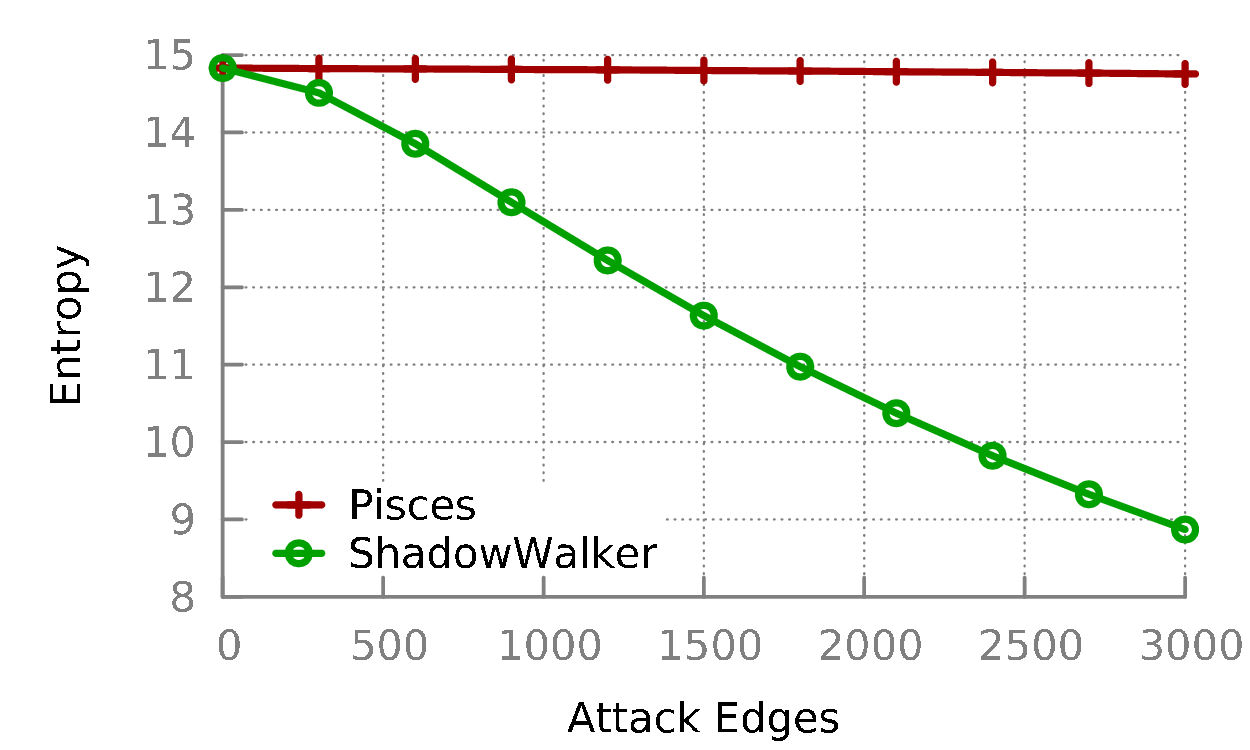,width=0.33 \textwidth}\vspace{-0.00in}\\
{(a)}
\end{tabular}
\begin{tabular}{c}
\psfig{figure=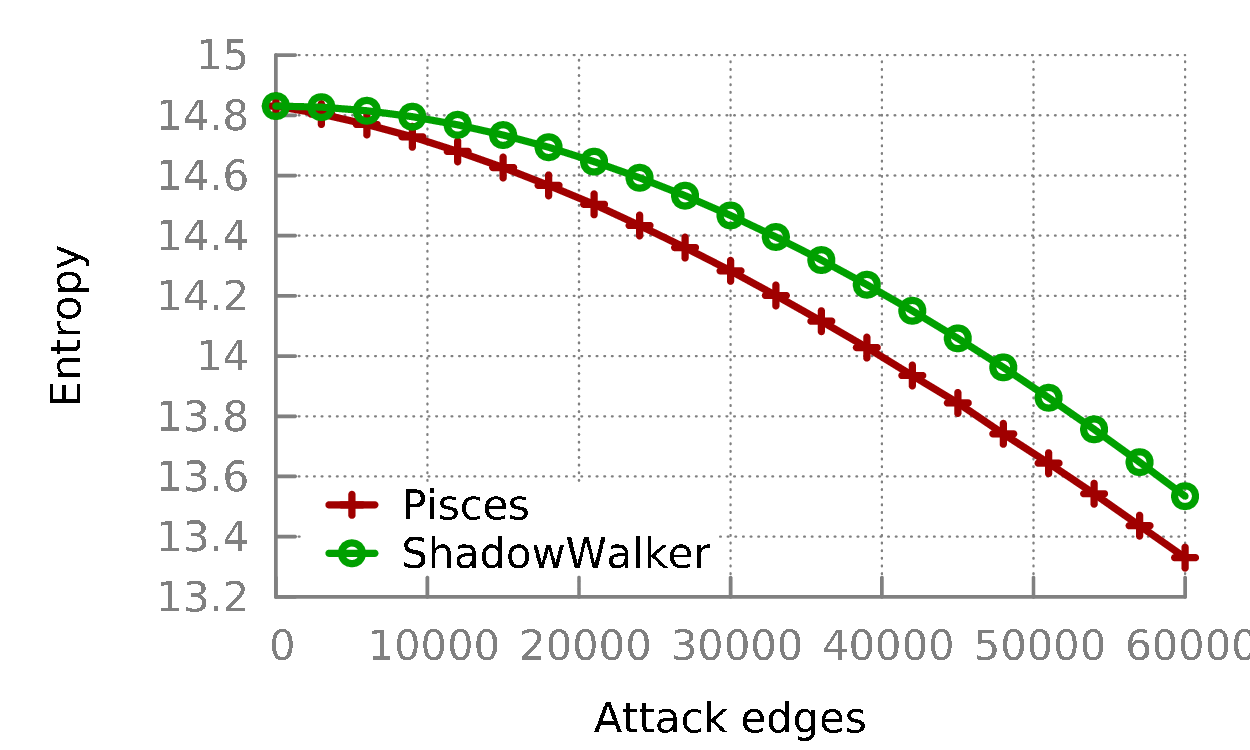,width=0.33 \textwidth}\vspace{-0.00in}\\
{(b)}
\end{tabular}
}
\begin{ccs}
\vspace{-5pt}
\end{ccs}
\caption{{\em Comparison with ShadowWalker. Entropy as a function of
    fraction of attack edges using} (a) realistic model of an imperfect
  Sybil defense (10 Sybils per attack edge) and (b) perfect Sybil
  defense for Facebook wall post interaction graph.  }
\label{fig:entropy-shadowwalker}
\end{figure*}
\end{techreport}
We now analyze the anonymity of the system with respect to an
insider adversary (malicious participants).%
We first
consider an adversary that has $g$ attack edges going to honest nodes,
with $g=O(\frac{h}{\log h})$, and 10 Sybils per attack
edge~\cite{sybillimit}. When both the first and the last hop of a random
walk are compromised, then initiator entropy is 0 due to end-to-end
timing analysis. Let $M_i$ be the event where the first compromised node
is at the $i$th hop and the last hop is also compromised. Suppose that
the previous hop of the first compromised node is node $A$. Under this
scenario, the adversary can localize the initiator to the set of nodes
that can reach the node $A$ in $i-1$ hops.  If we denote the initiator
anonymity under this scenario as $H(I|M_i)$, then from
equation~\eqref{eqn:system-entropy}, it follows that the overall system
anonymity is:
\begin{align}
H(I|O) = \sum_{i=1}^{i=l} P(M_i) \cdot H(I|M_i) + (1-\sum_{i=1}^{i=l} P(M_i)) \cdot \log_2 n
\end{align}
We compute $P(M_i)$ using simulations, and $H(I|M_i)$, using the
expected anonymity computations discussed above.
Figure~\ref{fig:entropy-malicious}(a) depicts the expected entropy as a
function of the number of attack edges. We find that Pisces provides
close to optimal anonymity. Moreover, as the length of the random walk
increases, the anonymity does not degrade. In contrast, without any
defense, the anonymity decreases with an increase in the random walk
length (not shown in the figure), since at every step in the random
walk, there is a chance of the random walk being captured by the
adversary. At $g=3000$, the anonymity provided by a conventional 10-hop
random walk without any defenses (used in systems such as Drac and Whanau) is $14.1$ bits, while Pisces provides
close to optimal anonymity at $14.76$ bit. For uniform probability
distributions, this represents an increase in the size of the anonymity
set by a factor of 1.6.
It is also interesting to see that the advantage of using Pisces
increases as the number of attack edge increases.
To further investigate this, we consider the attack model with perfect
Sybil defense and vary the number of attack edges.
Figure~\ref{fig:entropy-malicious}(b) depicts the anonymity as a
function of the number of attack edges. We can see that at 60\,000
attack edges, the expected anonymity without defenses is $7.5$ bits, as
compared to more than $13$ bits with Pisces (anonymity set size
increases by a factor of 45).

{\bf Comparison with ShadowWalker:} ShadowWalker~\cite{shadowwalker} is
a state-of-the-art approach for scalable anonymous communication that
organizes nodes into a structured topology such as Chord
and performs secure random walks on such topologies. We now compare our
approach with ShadowWalker.
To compute the anonymity provided by ShadowWalker, we use the fraction
$f$ of malicious nodes in the system as an input to the analytic model
of ShadowWalker~\cite{shadowwalker}, and use ShadowWalker parameters
that provide maximum security. Figure~\ref{fig:entropy-shadowwalker}(a)
depicts the comparative results between Pisces (using $l=25$) and
ShadowWalker. We can see that Pisces significantly outperforms
ShadowWalker. At $g=1000$ attack edges, Pisces provides about two bits
higher entropy than ShadowWalker, and this difference increases to six 
bits at $g=3000$ attack edges\footnote{At such high attack edges,
  ShadowWalker may even have difficulty in securely maintaining its
  topology, which could further lower anonymity.}.  This difference
arises because Pisces directly performs random walks on the social
network topology, limiting the impact of Sybil attackers, while
ShadowWalker is designed to secure random walks only on structured
topologies. Arranging nodes in a structured topology loses information
about trust relationships between users, resulting in poor anonymity for
ShadowWalker. \footnote{This observation is also applicable to Tor. Pisces 
provides 5 bits higher entropy than Tor at $g=3000$ attack edges.}
For comparison, we also consider the attack model with perfect Sybil
defense and vary the number of attack
edges. Figure~\ref{fig:entropy-shadowwalker}(b) depicts the results for
this scenario. We can see that even in this scenario where trust
relationships lose meaning since the adversary is randomly distributed,
Pisces continues to provides comparable anonymity to
ShadowWalker. Pisces's entropy is slightly lower since social networks
are slower mixing than structured networks, requiring longer length
random walks than ShadowWalker and thereby giving more observation
points to the adversary.

{\bf Performance optimization:} 
We now analyze the anonymity provided by our two hop optimization, which
uses the $k$-th hop and the last hop of the random walk for anonymous
communication. To analyze anonymity in this scenario, let us redefine
$M_i$ ($i\neq k$) as the event when the first compromised node is at the
$i$-th hop, the last node is also compromised, but the $k$-th node is
honest. Let $M_k$ be the event where the $k$-th hop and the last hop are
compromised (regardless of whether other nodes are compromised or not)
and the definition of $M_l$ remains the same as before, i.e., only the
last hop is compromised. We can compute system anonymity as:
\begin{align}
H(I|O)  & = \sum_{i=1}^{i=k-1} P(M_i) \cdot H(I|M_i) + \sum_{i=k+1}^{l} P(M_i) \cdot H(I|M_k) \nonumber \\ 
   & + (1-\sum_{i=1}^{i=l} P(M_i)) \cdot \log_2 n
\end{align}
\begin{ccs}
\begin{figure}[tph]
\centering
\includegraphics[width=2.5in]{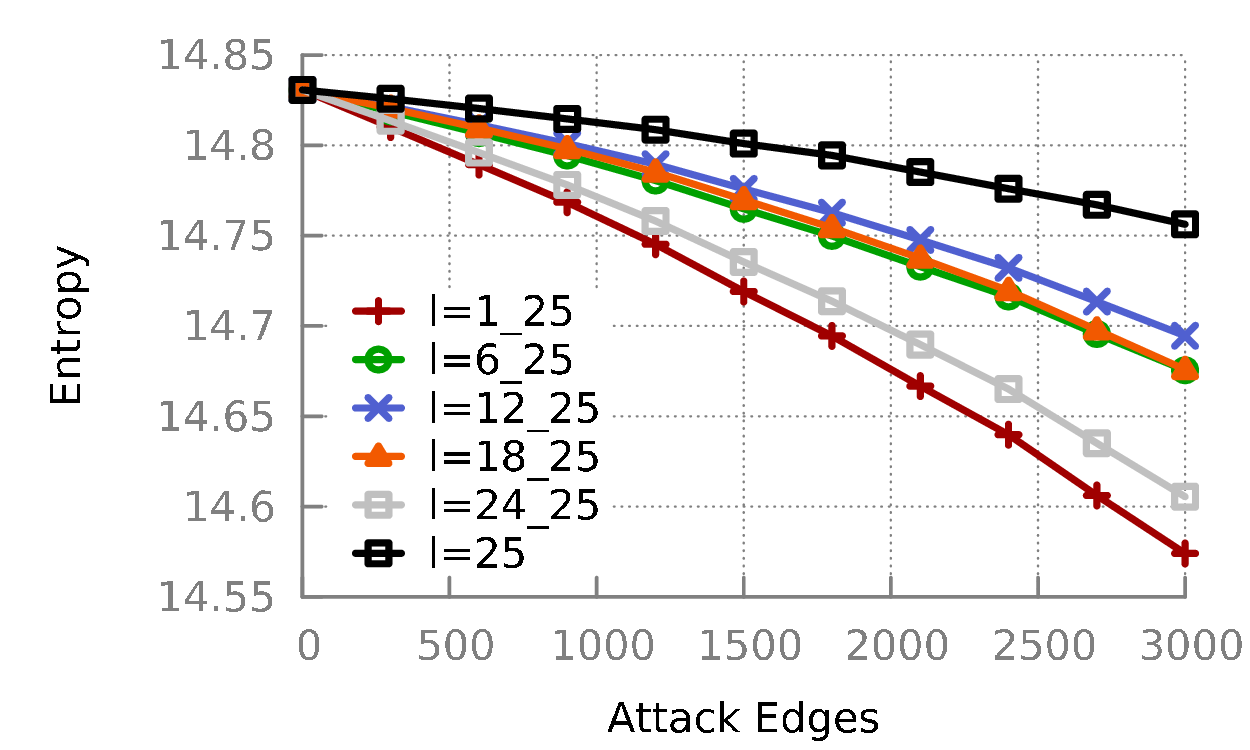}\\
\begin{ccs}
\vspace{-10pt}
\end{ccs}
\caption{Anonymity using the two hop performance optimization, Facebook
  wall graph, 10 Sybils/attack edge. $k=12$ results in provides a good
  trade-off between anonymity and performance.}
\label{fig:two-hop}
\vspace{-10pt}
\end{figure}
\end{ccs}
Figure~\ref{fig:two-hop} depicts the anonymity for our two hop
optimization for different choices of $k$. We see an interesting
trade-off here. Small values of $k$ are not optimal, since even though
the first hop is more likely to be honest, when the last hop is
compromised, then the initiator is easily localized. On the other hand,
large values of $k$ are also not optimal, since these nodes are far away
from the initiator in the trust graph and are less trusted.  We find
that optimal trade-off points are in the middle, with $k=12$ providing
the best anonymity for our optimization. We also note that the anonymity
provided by our two hop optimization is close to the anonymity provided
by using all 25 hops of the random walk.

{\bf Selective denial of service:}
\label{sec:seldos}
\begin{techreport}
\begin{figure}[tph]
\centering
\includegraphics[width=2.5in]{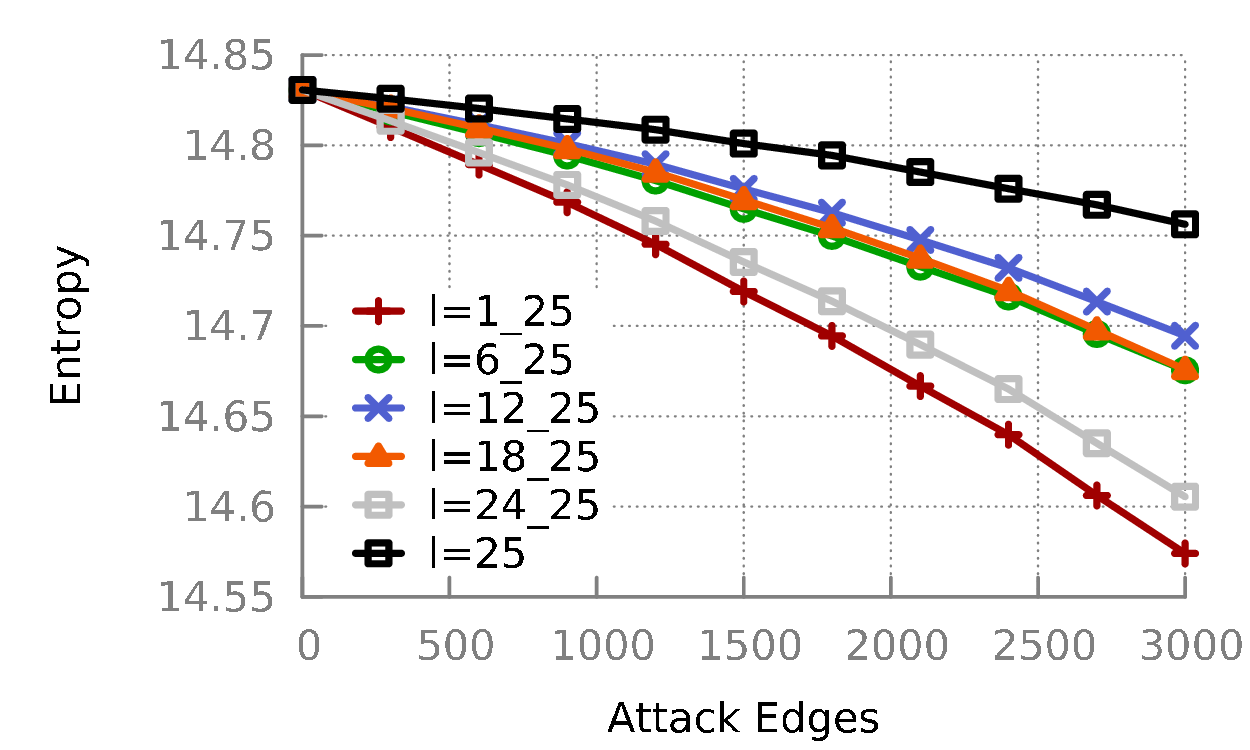}\\
\caption{Anonymity using the two hop performance optimization, Facebook
  wall graph, 10 Sybils/attack edge. $k=12$ provides a good trade-off between
  anonymity and performance.}
\label{fig:two-hop}
\end{figure}
\end{techreport}
\begin{figure}[tph]
\centering
\includegraphics[width=2.5in]{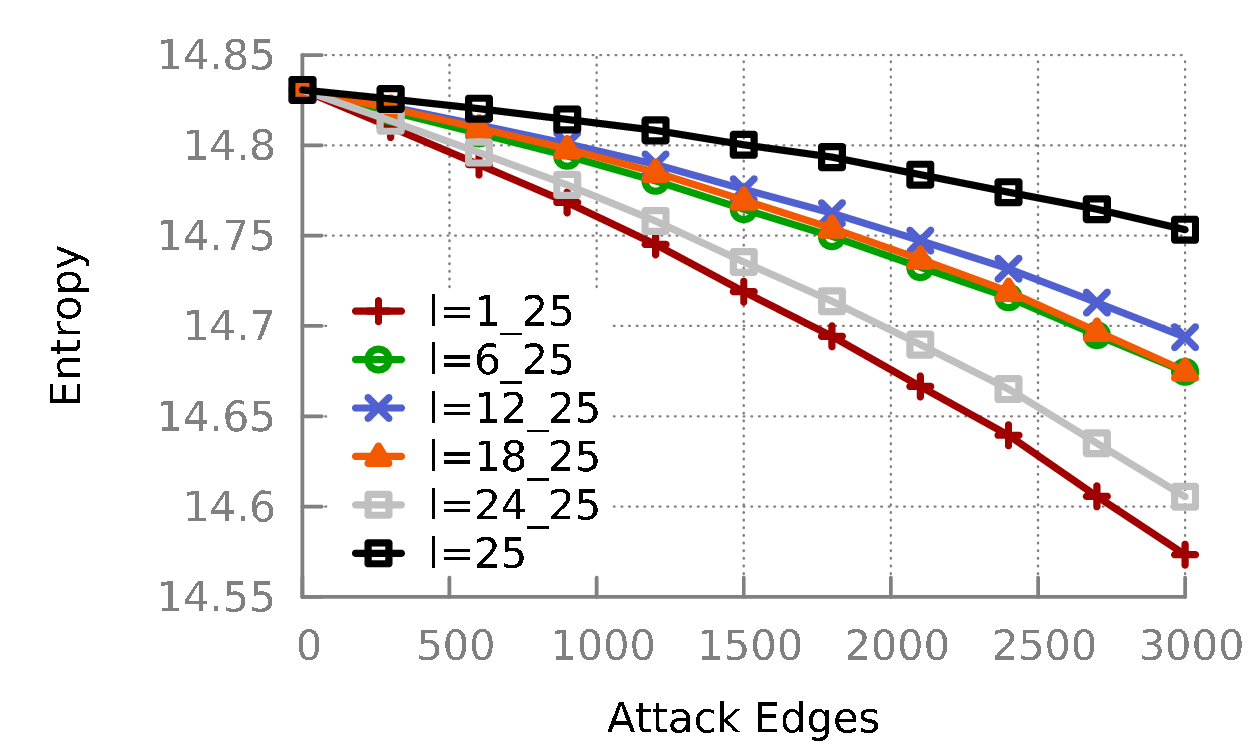}\\
\caption{Anonymity under the Selective DoS attack using the Facebook
  wall graph, 10 Sybils/attack edge. Selective DoS has limited impact.}
\label{fig:selective-dos}
\end{figure}
Next, we evaluate Pisces anonymity against the selective DoS
attack~\cite{borisov:ccs07}. In this attack, an adversary can cause a
circuit to selectively fail whenever he or she is unable to learn the
initiator identity. This forces the user to construct another circuit,
which results in a degradation of anonymity.
We found that the degradation in initiator anonymity under this attack 
is less than 1\%. The reason why Pisces is less vulnerable to selective DoS 
as compared with other systems such as Tor is due to the use of social trust. With high 
probability, most random walks traverse only the honest set of nodes. 
This result is illustrated in Figure~\ref{fig:selective-dos}. %
{\bf Multiple communication rounds:} 
So far, we had limited our analysis to a single communication
round. Next, we analyze system anonymity over multiple communication
rounds.
Let us suppose that in communication rounds $1 \ldots z$, the
adversary's observations are $O_1 \ldots O_z$. Let us denote a given
node's probability of being the initiator after $z$ communication rounds
by $P(I=i|O_1,\ldots ,O_z)$. Now, after communication round $z+1$, we
are interested in computing the probability $P(I=i|O_1,\ldots
,O_{z+1})$. Using Bayes's theorem, we have that:
\begin{align}
P(I=i|O_1, \ldots , O_{z+1}) & = \frac{P(O_1, \ldots , O_{z+1}| I=i) \cdot P(I=i)}{P(O_1, \ldots , O_{z+1})}
\end{align}
The key advantage of this formulation is that we can now leverage the
observations $O_1, \ldots O_{z+1}$ being independent given a choice of
initiator. Thus we have that:
\begin{align}
P(I=i|O_1, & \ldots , O_{z+1})  = \frac{ \Pi_{j=1}^{j=z+1} P(O_j| I=i) \cdot P(I=i)}{P(O_1, \ldots , O_{z+1})} \nonumber \\
  & =   \frac{ \Pi_{j=1}^{j=z+1} P(O_j| I=i) \cdot P(I=i)}{ \sum_{p=1}^{p=h} P(O_1, \ldots , O_{z+1} | I=p) \cdot P(I=p)} \nonumber \\
  & = \frac{ \Pi_{j=1}^{j=z+1} P(O_j| I=i) \cdot P(I=i)}{ \sum_{p=1}^{p=h} \Pi_{j=1}^{j=z+1} P(O_j| I=p) \cdot P(I=p)} 
\end{align}
Finally, assuming a uniform prior over all possible initiators, we have
that:
\begin{align}
P(I=i|O_1, \ldots , O_{z+1}) & = \frac{ \Pi_{j=1}^{j=z+1} P(O_j| I=i)}{ \sum_{p=1}^{p=h} \Pi_{j=1}^{j=z+1} P(O_j| I=p)} 
\end{align}
\begin{figure}[t]
\centering
\includegraphics[width=2.5in]{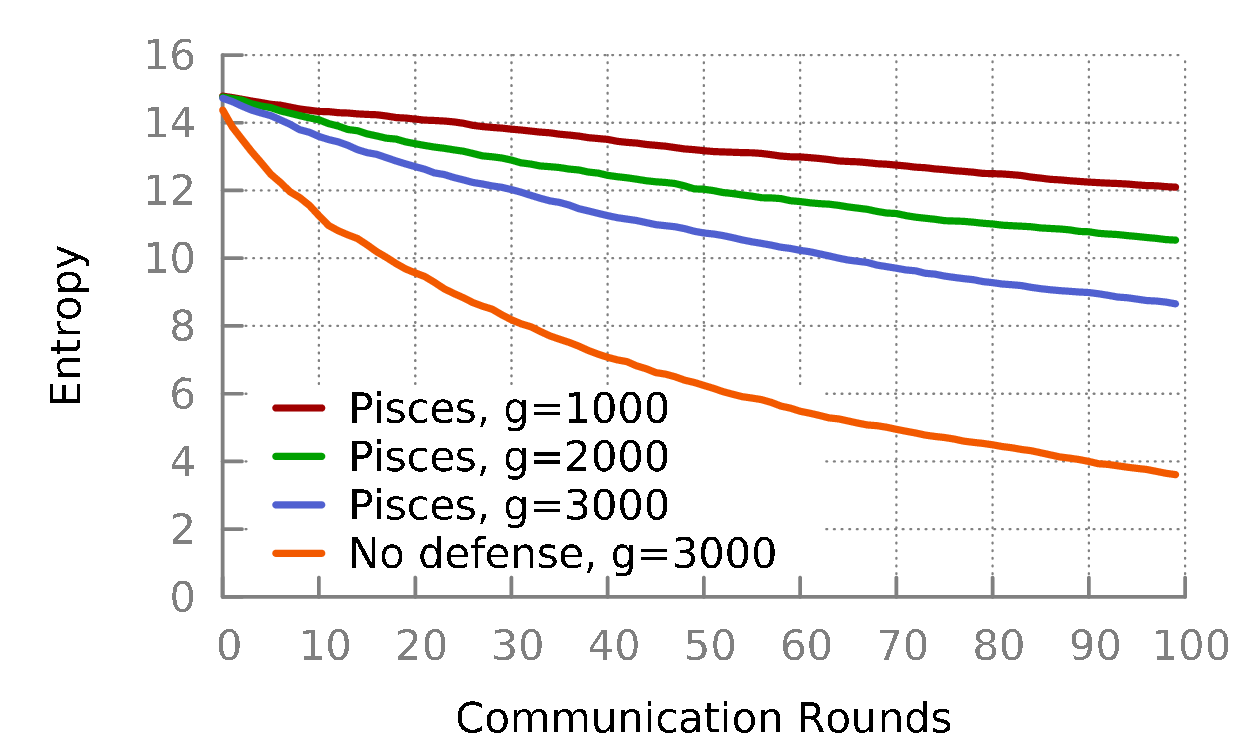}\\
\begin{ccs}
\vspace{-10pt}
\end{ccs}
\caption{Anonymity degradation over multiple communication rounds,
  Facebook wall graph, 10 Sybils/attack edge}
\label{fig:multi-round}
\begin{ccs}
\vspace{-10pt}
\end{ccs}
\end{figure}
Figure~\ref{fig:multi-round} depicts the expected anonymity as a
function of number of communication rounds. We can see that the entropy
provided by Pisces outperforms conventional random walks by more than a
factor of two (in bits) after 100 communication rounds (the anonymity 
set size is increased by a factor of 16).
\begin{ccs}
\section{Limitations and Future Work}
\end{ccs}

\begin{techreport}
\section{Discussion}
{\em Integration with Sybil defenses:} We outline two strategies 
for integrating Pisces with Sybil defense. The first approach 
is to leverage mechanisms that require the whole social graph as 
input for Sybil defense, such as SybilInfer. This has the downside 
of communication overhead for reliably maintaining the 
social graph accurately, in presence of social graph churn such 
as new users and new trust relationships. The upside of this approach 
is that \emph{while} performing random walks, both for anonymous 
communication and for testing, no further communication is 
required to validate identities for Sybil defense. The second 
approach is to leverage decentralized mechanisms like SybilLimit 
that do not require users to maintain global information about the 
social graph. However, in this scenario, while performing random walks, 
each hop of the random walk must be validated for Sybil defense. A key 
challenge in validating nodes while performing random walks is to prevent 
other entities in the network from learning the nodes involved in the 
random walk performed by an initiator. Towards this end, we propose that 
the node being validated return its list of SybilLimit \emph{tails} to 
the initiator using the partial onion routing circuit, who can then 
perform a privacy preserving set intersection protocol with its tails to 
perform Sybil defense.  
\end{techreport}

\begin{techreport}
{\em Limitations:} 
\end{techreport}

While Pisces is the first decentralized design that can both scalably
leverage social network trust relationships and mitigate route
capture attacks, its architecture has some limitations. First, Pisces
requires user's social contacts to participate in the system. To improve 
the usability of the system, in future
work, we will investigate the feasibility of leveraging a user's two-hop
social neighborhood in the random walk process. %
Pisces also does not preserve the privacy of users' social contacts. 
We emphasize that this is also a limitation shared by a large class of social 
network based systems, including Sybil defense mechanisms such as SybilLimit, 
SybilInfer and Whanau. Any distributed anonymity system relying on such 
Sybil defense mechanisms cannot preserve the privacy of users' social 
contacts. 

Second, users who are not well connected in the social
network topology may not benefit from using Pisces. This is because
random walks starting from those nodes may take a very long time to
converge to the stationary probability distribution (which provides
optimal anonymity). 

Third, Pisces does not defend against targeted
attacks on an individual, in which the adversary aims to massively
infiltrate or compromise the user's social circle for increasing the
probability of circuit compromise. We note that the impact of such an
attack is localized to the targeted individual. 

Fourth, circuit
establishment in Pisces has higher latency than existing systems, since
random walks in Pisces tend to be longer. However, we note that
circuits can be established pre-emptively, such that this latency does
not affect the user. In fact, deployed systems such as Tor already build
circuits pre-emptively. 

Finally, Pisces currently does not support
important constraints such as bandwidth-based load balancing and exit
policies. The focus of our architecture was to secure the peer discovery
process in unstructured social network topologies, and we will consider
the incorporation of these constraints in future work.

\section{Conclusion}
\label{sec:conclusion}

In this paper, we propose a mechanism for decentralized anonymous
communication that can securely leverage a user's trust relationships
against a Byzantine adversary. Our key contribution is to show that
appearance of nodes in each other's neighbor lists can be made
reciprocal in a secure and efficient manner. Using theoretical analysis
and experiments on real world social network topologies, we demonstrate
that Pisces substantially reduces the probability of active attacks on
circuit constructions.
We find that Pisces significantly outperforms approaches that do not
leverage trust relationships, and provides up to six bits higher entropy 
than ShadowWalker (5 bits higher entropy than Tor) in a single communication
round. Also, compared with the naive strategy of using conventional
random walks over social networks (as in the Drac system), Pisces
provides twice the number of bits of entropy over 100 communication
rounds. In conclusion, we argue that the incorporation of social trust
will likely be an important consideration in the design of the next
generation of deployed anonymity systems.  \balance

\section*{Acknowledgment}
We would like to thank the attendees of HotPETs 2010 and USENIX HotSec
2010 for helpful comments. In particular, this work benefited from
conversations with George Danezis, Aaron Johnson, and Paul
Syverson. This work is sponsored in part by NSF CAREER Award, number
CNS-0954133, and by award number CNS-1117866. Any opinions, findings and
conclusions or recommendations expressed in this material are those of
the author(s) and do not necessarily reflect those of the National
Science Foundation.

{
\bibliographystyle{abbrv}
\bibliography{refs,scaling-anonymity-cosic}
}


\end{document}